%% file: cm.tex
\documentclass[twoside, 11pt]{article}
\usepackage{amsmath}
\usepackage{jmlr2e}

\usepackage[T1]{fontenc}
\usepackage{lmodern}
\usepackage{graphicx}
\usepackage{booktabs}
\usepackage{amssymb}
\usepackage{algorithmic}
\usepackage{algorithm}
\usepackage{url}

\newcommand{\set}[1]{\left\{#1\right\}}
\newcommand{\iset}[1]{\mathcal{#1}}
\newcommand{\pr}[1]{\left(#1\right)}
\newcommand{\spr}[1]{\left[#1\right]}
\newcommand{\abs}[1]{\left|#1\right|}
\newcommand{\norm}[1]{\left\|#1\right\|_2}
\newcommand{\enset}[2]{\set{#1 ,\ldots , #2}}
\newcommand{\enpr}[2]{\pr{#1 ,\ldots , #2}}
\newcommand{\vect}[1]{\spr{#1}^T}
\newcommand{\sqrta}[1]{\spr{#1}^{1/2}}

\newcommand{\envec}[2]{\vect{#1 ,\ldots , #2}}
\newcommand{\mean}[2]{\operatorname{E}_{#1}\spr{#2}}
\newcommand{\cov}[3]{\operatorname{Cov}_{#1}^{#2}\spr{#3}}
\newcommand{\cova}[3]{\operatorname{Cov}^{#2}\spr{{#3}\mid{{#1}}}}
\newcommand{\argmin}{\operatorname{argmin}}
\newcommand{\real}{\mathbb{R}}
\newcommand{\kl}[2]{D\pr{#1 \| #2}}
\newcommand{\alname}[1]{\textsc{#1}}
\newcommand{\dtname}[1]{\emph{#1}}
\newcommand{\ftname}[1]{\emph{#1}}

\newcommand{\funcdef}[3]{{#1}:{#2} \to {#3}}

\newcommand{\const}[2]{\mathcal{C}\pr{#1, #2}}
\newcommand{\constp}[2]{\mathcal{C_+}\pr{#1, #2}}
\newcommand{\freq}[2]{{#1}\pr{#2}}

\newcommand{\dcm}[1]{d_{CM}\pr{#1}}
\newcommand{\du}[1]{d_U\pr{#1}}

\newcommand{\dist}[3]{d_{CM}\pr{#1, #2 \mid #3}}
\newcommand{\distu}[3]{d_{U}\pr{#1, #2 \mid #3}}
\newcommand{\dista}[3]{d\pr{#1, #2 \mid #3}}

\jmlrheading{X}{XXX}{XX}{2/05; Revised 9/05}{XXX}{Nikolaj Tatti}

\ShortHeadings{Distances between Data Sets Based on Summary Statistics}{Tatti}
\firstpageno{1}

\title{Distances between Data Sets Based on Summary Statistics}
\author{\name Nikolaj Tatti  \email ntatti@cc.hut.fi \\
\addr HIIT Basic Research Unit\\
Laboratory of Computer and Information Science\\
Helsinki University of Technology, Finland}
\editor{XXX}
\begin{document}
\maketitle
\begin{abstract}
The concepts of similarity and distance are crucial in data mining. We consider the problem of defining the distance between two data sets by comparing summary statistics computed from the data sets. The initial definition of our distance is based on geometrical notions of certain sets of distributions. We show that this distance can be computed in cubic time and that it has several intuitive properties. We also show that this distance is the unique Mahalanobis distance satisfying certain assumptions. We also demonstrate that if we are dealing with binary data sets, then the distance can be represented naturally by certain parity functions, and that it can be evaluated in linear time. Our empirical tests with real world data show that the distance works well.
\end{abstract}
\input{introduction.tex}
\input{theory.tex}
\input{binary.tex}
\input{sequence.tex}
\input{feature.tex}
\input{related.tex}
\input{tests.tex}

\input{conclusions.tex}
\acks{The author would like to thank Heikki Mannila and Kai Puolam\"{a}ki for their extremely helpful comments.}
\bibliography{bibliography}
\appendix
\input{appendix.tex}
\end{document}

%% file: introduction.tex
\section{Introduction}
In this paper we will consider the following problem: Given two data sets $D_1$ and $D_2$ of dimension $K$, define a distance between $D_1$ and $D_2$. To be more precise, we consider the problem of defining the distance between two multisets of transactions, each set sampled from its own unknown distribution. We will define a dissimilarity measure between $D_1$ and $D_2$ and we will refer to this measure as \emph{CM distance}.

Generally speaking, the notion of dissimilarity between two objects is one of the most fundamental concepts in data mining. If one is able to retrieve a distance matrix from a set of objects, then one is able to analyse data by using e.g., clustering or visualisation techniques. Many real world data collections may be naturally divided into several data sets. For example, if a data collection consists of movies from different eras, then we may divide the movies into subcollections based on their release years. A distance between these data (sub)sets would provide means to analyse them as single objects. Such an approach may ease the task of understanding complex data collections.

Let us continue by considering the properties the CM distance should have. First of all, it should be a metric. The motivation behind this requirement is that the metric theory is a well-known area and metrics have many theoretical and practical virtues. Secondly, in our scenario the data sets have statistical nature and the CM distance should take this into account. For example, consider that both data sets are generated from the same distribution, then the CM distance should give small values and approach $0$ as the number of data points in the data sets increases. The third requirement is that we should be able to evaluate the CM distance quickly. This requirement is crucial since we may have high dimensional data sets with a vast amount of data points.

The CM distance will be based on summary statistics, features. Let us give a simple example: Assume that we have data sets $D_1 = \set{A, B, A, A}$ and $D_2 = \set{A, B, C, B}$ and assume that the only feature we are interested in is the proportion of $A$ in the data sets. Then we can suggest the distance between $D_1$ and $D_2$ to be $\abs{3/4-1/4} = 1/2$. The CM distance is based on this idea; however, there is a subtle difficulty: If we calculate several features, then should we take into account the correlation of these features? We will do exactly that in defining the CM distance.

The rest of this paper is organised as follows. In Section~\ref{sec:cm} we give the definition of the CM distance by using some geometrical interpretations. We also study the properties of the distance and provide an alternative characterisation. In Section~\ref{sec:cmbin} we study the CM distance and binary data sets. In Section~\ref{sec:sequences} we discuss how the CM distance can be used with event sequences and in Section~\ref{sec:feature} we comment about the feature selection. Section~\ref{sec:related} is devoted for related work. The empirical tests are represented in Section~\ref{sec:tests} and we conclude our work with the discussion in Section~\ref{sec:conclusions}.

%% file: theory.tex
\section{The Constrained Minimum Distance}
\label{sec:cm}
In the following subsection we will define our distance using geometrical intuition and show that the distance can be evaluated efficiently. In the second subsection we will discuss various properties of the distance, and in the last subsection we will provide an alternative justification to the distance. The aim of this justification is to provide more theoretical evidence for our distance.
\subsection{The definition}
We begin by giving some basic definitions. By a \emph{data set} $D$ we mean a finite collection of samples lying in some finite space $\Omega$. The set $\Omega$ is called \emph{sample space}, and from now on we will denote this space by the letter $\Omega$. The number of elements in $\Omega$ is denoted by $\abs{\Omega}$. The number of samples in the data set $D$ is denoted by $\abs{D}$.

As we said in the introduction, our goal is not to define a distance directly on data sets but rather through some statistics evaluated from the data sets. In order to do so, we define a \emph{feature function} $\funcdef{S}{\Omega}{\real^N}$ to map a point in the sample space to a real vector. Throughout this section $S$ will indicate some given feature function and $N$ will indicate the dimension of the range space of $S$. We will also denote the $i^{\text{th}}$ component of $S$ by $S_i$. Note that if we have several feature functions, then we can join them into one big feature function. A \emph{frequency} $\theta \in \real^N$ of $S$ taken with respect to a data set $D$ is the average of values of $S$ taken over the data set, that is, $\theta = \frac{1}{\abs{D}}\sum_{\omega \in D} S(\omega)$. We denote this frequency by $\freq{S}{D}$.

Although we do not make any assumptions concerning the size of $\Omega$, some of our choices are motivated by thinking that $\abs{\Omega}$ can be very large --- so large that even the simplest operation, say, enumerating all the elements in $\Omega$, is not tractable. On the other hand, we assume that $N$ is such that an algorithm executable in, say, $O(N^3)$ time is feasible. In other words, we seek a distance whose evaluation time does not depend of the size of $\Omega$ but rather of $N$.

Let $\mathbb{P}$ be the set of all distributions defined on $\Omega$. Given a feature function $S$ and a frequency $\theta$ (calculated from some data set) we say that a distribution $p \in \mathbb{P}$ satisfies the frequency $\theta$ if $\mean{p}{S} = \theta$. We also define a \emph{constrained set of distributions}
\[
\constp{S}{\theta} = \set{p \in \mathbb{P} \mid \mean{p}{S} = \theta}
\]
to be the set of the distributions satisfying $\theta$. The idea behind this is as follows: From a given data set we calculate some statistics, and then we examine the distributions that can produce such frequencies.

We interpret the sets $\mathbb{P}$ and $\constp{S}{\theta}$ as \emph{geometrical objects}. This is done by enumerating the points in $\Omega$, that is, we think that $\Omega = \enset{1,2,}{\abs{\Omega}}$. We can now represent each distribution $p \in \mathbb{P}$ by a vector $u \in \real^{\abs{\Omega}}$ by setting $u_i = p(i)$. Clearly, $\mathbb{P}$ can be represented by the vectors in $\real^{\abs{\Omega}}$ having only non-negative elements and summing to one. In fact, $\mathbb{P}$ is a simplex in $\real^{\abs{\Omega}}$. Similarly, we can give an alternative definition for $\constp{S}{\theta}$ by saying
\begin{equation}
\constp{S}{\theta} = \set{u \in \real^{\abs{\Omega}} \mid \sum_{i \in \Omega }S(i)u_i = \theta, \sum_{i \in \Omega }u_i = 1, u \geq 0}.
\label{eq:constpdef}
\end{equation}
Let us now study the set $\constp{S}{\theta}$. In order to do so, we define a \emph{constrained space}
\[
\const{S}{\theta} = \set{u \in \real^{\abs{\Omega}} \mid \sum_{i \in \Omega }S(i)u_i = \theta, \sum_{i \in \Omega }u_i = 1},
\]
that is, we drop the last condition from Eq.~\ref{eq:constpdef}. The set $\constp{S}{\theta}$ is included in $\const{S}{\theta}$; the set $\constp{S}{\theta}$ consists of the non-negative vectors from $\const{S}{\theta}$. Note that the constraints defining $\const{S}{\theta}$ are vector products. This implies that $\const{S}{\theta}$ is an affine space, and that, given two different frequencies $\theta_1$ and $\theta_2$, the spaces $\const{S}{\theta_1}$ and $\const{S}{\theta_2}$ are parallel. 
\begin{example}
Let us illustrate the discussion above with a simple example. Assume that $\Omega = \set{A,B,C}$. We can then imagine the distributions as vectors in $\real^3$. The set $\mathbb{P}$ is the triangle having $\pr{1, 0, 0}$, $\pr{0, 1, 0}$, and $\pr{0, 0, 1}$ as corner points (see Figure~\ref{fig:ex1plot}). Define a feature function $S$ to be
\[
S(\omega) = \left\{
\begin{array}{ll}
1 & \omega = C \\
0 & \omega \neq C.
\end{array}\right.
\]
The frequency $\freq{S}{D}$ is the proportion of $C$ in a data set $D$. Let $D_1 = \pr{C, C, C, A}$ and $D_2 = \pr{C, A, B, A}$. Then $\freq{S}{D_1} = 0.75$ and $\freq{S}{D_2} = 0.25$. The spaces $\const{S}{0.25}$ and $\const{S}{0.75}$ are parallel lines (see Figure~\ref{fig:ex1plot}). The distribution sets $\constp{S}{0.25}$ and $\constp{S}{0.75}$ are the segments of the lines $\const{S}{0.25}$ and $\const{S}{0.75}$, respectively.
\begin{figure}[ht!]
\center
\begin{minipage}{7cm}
\includegraphics[width=7cm]{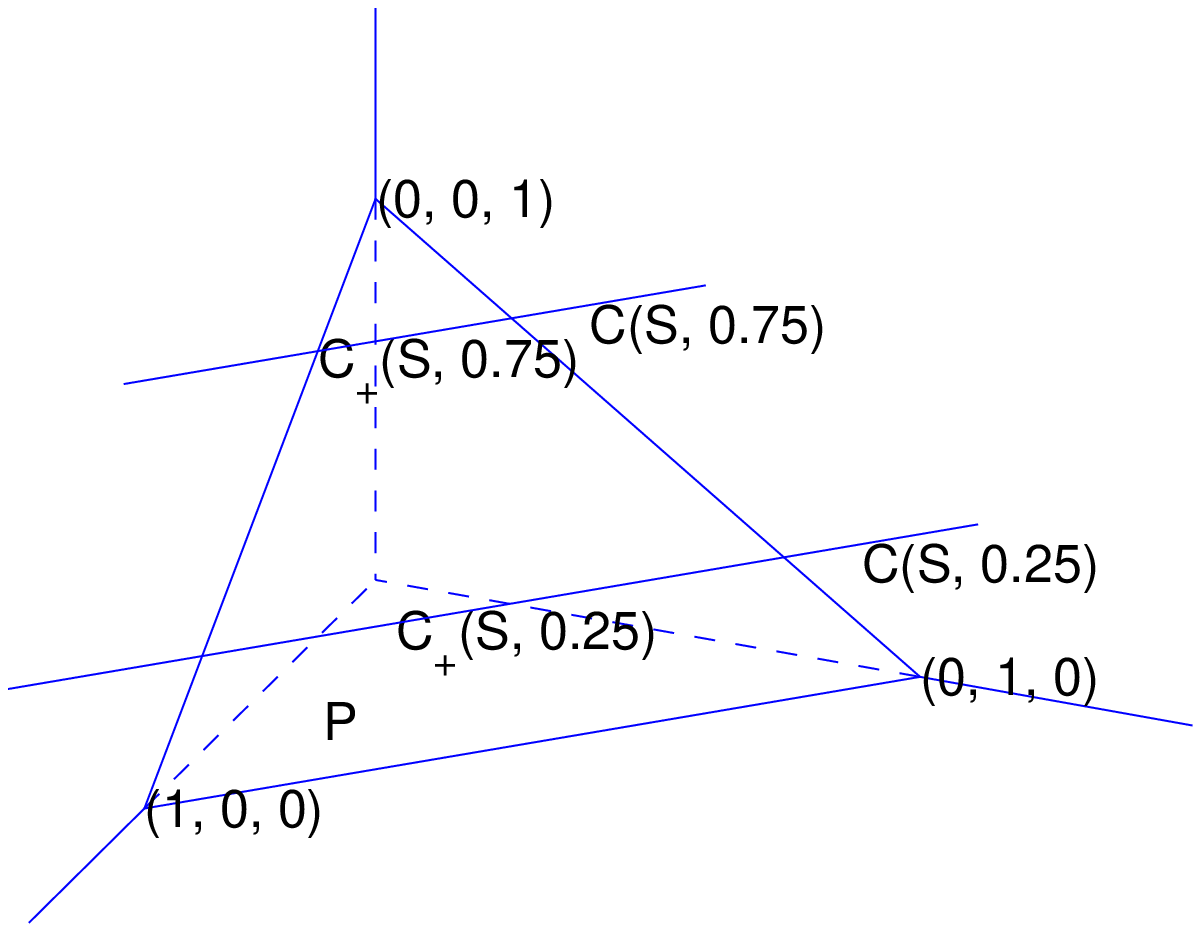}
\end{minipage}
\begin{minipage}{7cm}
\includegraphics[width=7cm]{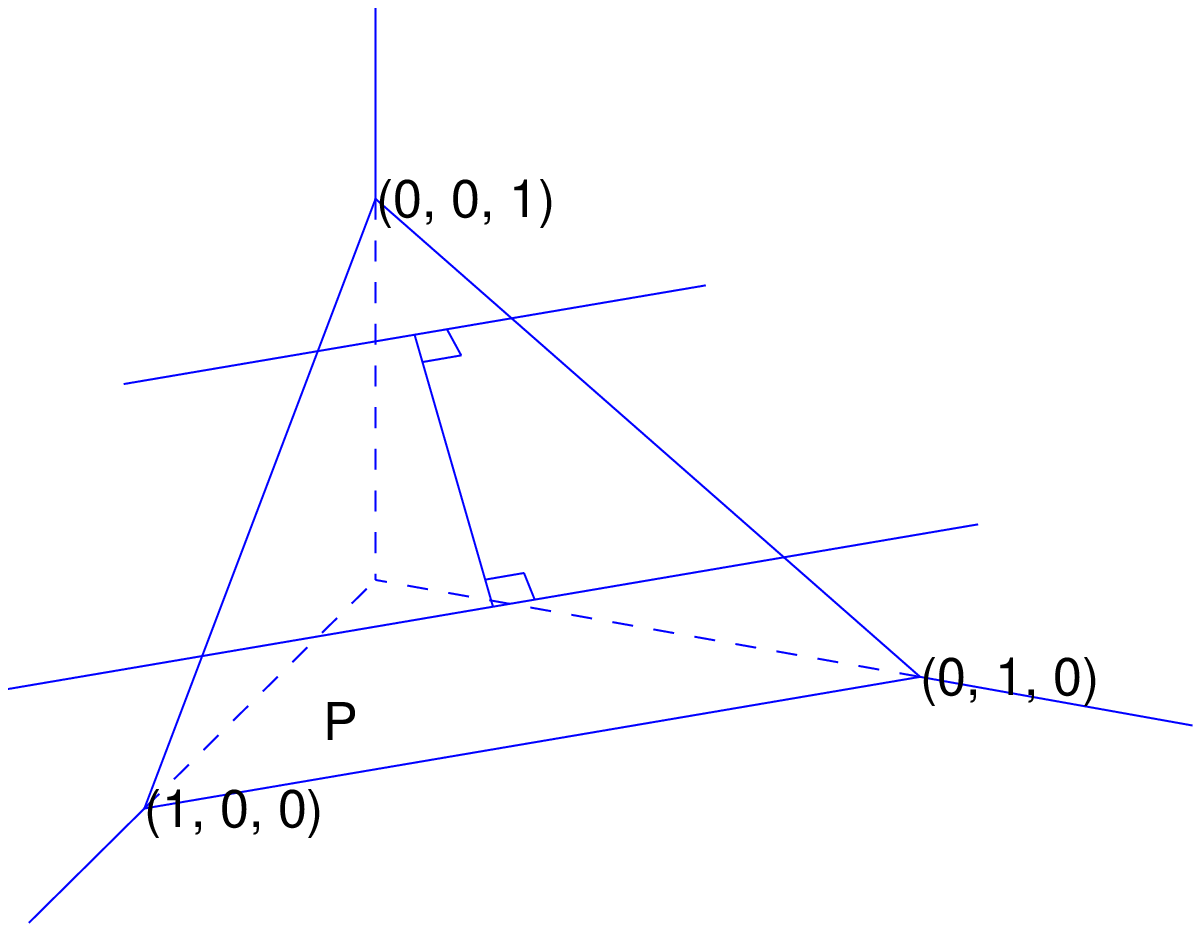}
\end{minipage}
\caption{A geometrical interpretation of the distribution sets for $\abs{\Omega} = 3$. In the left figure, the set $\mathbb{P}$, that is, the set of all distributions, is a triangle. The constrained spaces $\const{S}{0.25}$ and $\const{S}{0.75}$ are parallel lines and the distribution sets $\constp{S}{0.25}$ and $\constp{S}{0.75}$ are segments of the constrained spaces. In the right figure we added a segment perpendicular to the constraint spaces. This segment has the shortest length among the segments connecting the constrained spaces.}
\label{fig:ex1plot}
\end{figure}
\label{ex:illustration}
\end{example}
The idea of interpreting distributions as geometrical objects is not new. For example, a well-known boolean query problem  is solved by applying linear programming to the constrained sets of distributions~\citep{hailperin65inequalities, calders03thesis}.

Let us revise some elementary Euclidean geometry: Assume that we are given two parallel affine spaces $\mathcal{A}_1$ and $\mathcal{A}_2$. There is a natural way of measuring the distance between these two spaces. This is done by taking the length of the shortest segment going from a point in $\mathcal{A}_1$ to a point in $\mathcal{A}_2$ (for example see the illustration in Figure~\ref{fig:ex1plot}). We know that the segment has the shortest length if and only if it is orthogonal with the affine spaces. We also know that if we select a point $a_1 \in \mathcal{A}_1$ having the shortest norm, and if we similarly select $a_2 \in \mathcal{A}_2$, then the segment going from $a_1$ to $a_2$ has the shortest length.

The preceeding discussion and the fact that the constrained spaces are affine motivates us to give the following definition: Assume that we are given two data sets, namely $D_1$ and $D_2$ and a feature function $S$. Let us shorten the notation $\const{S}{\freq{S}{D_i}}$ by $\const{S}{D_i}$. We pick a vector from each constrained space having the shortest norm
\[
u_i = \underset{u \in \const{S}{D_i}}{\argmin} \norm{u}, \quad i = 1,2.
\]
We define the distance between $D_1$ and $D_2$ to be
\begin{equation}
\dist{D_1}{D_2}{S} = \sqrt{\abs{\Omega}}\norm{u_1-u_2}.
\label{eq:def}
\end{equation}
The reasons for having the factor $\sqrt{\abs{\Omega}}$ will be given later. We will refer to this distance as \emph{Constrained Minimum (CM) distance}. We should emphasise that $u_1$ or $u_2$ may have negative elements. Thus the CM distance is \emph{not} a distance between two distributions; it is rather a distance based on the frequencies of a given feature function and is motivated by the geometrical interpretation of the distribution sets.

The main reason why we define the CM distance using the constrained spaces $\const{S}{D_i}$ and not the distribution sets $\constp{S}{D_i}$ is that we can evaluate the CM distance efficiently. We discussed earlier that $\Omega$ may be very large so it is crucial that the evaluation time of a distance does not depend on $\abs{\Omega}$. The following theorem says that the CM distance can be represented using the frequencies and a covariance matrix
\[
\cov{}{}{S} = \frac{1}{\abs{\Omega}}\sum_{\omega \in \Omega}S(\omega)S(\omega)^T-\pr{\frac{1}{\abs{\Omega}}\sum_{\omega \in \Omega}S(\omega)}\pr{\frac{1}{\abs{\Omega}}\sum_{\omega \in \Omega}S(\omega)}^T.
\]
\begin{theorem}
\label{thr:calc}
Assume that $\cov{}{}{S}$ is invertible. For the CM distance between two data sets $D_1$ and $D_2$ we have
\[
\dist{D_1}{D_2}{S}^2 = \pr{\theta_1-\theta_2}^T\cov{}{-1}{S}\pr{\theta_1-\theta_2},
\]
where $\theta_i = \freq{S}{D_i}$.
\end{theorem}
The proofs for the theorems are given in Appendix.

The preceding theorem shows that we can evaluate the distance using the covariance matrix and frequencies. If we assume that evaluating a single component of the feature function $S$ is a unit operation, then the frequencies can be calculated in $O(N\abs{D_1}+N\abs{D_2})$ time. The evaluation time of the covariance matrix is $O(\abs{\Omega}N^2)$ but we assume that $S$ is such that we know a closed form for the covariance matrix (such cases will be discussed in Section~\ref{sec:cmbin}), that is, we assume that we can evaluate the covariance matrix in $O(N^2)$ time. Inverting the matrix takes $O(N^3)$ time and evaluating the distance itself is $O(N^2)$ operation. Note that calculating frequencies and inverting the covariance matrix needs to be done only once: for example, assume that we have $k$ data sets, then calculating the distances between every data set pair can be done in $O\pr{N\sum_i^k\abs{D_i}+N^3+k^2N^2}$ time.
 
\begin{example}
Let us evaluate the distance between the data sets given in Example~\ref{ex:illustration} using both the definition of the CM distance and Theorem~\ref{thr:calc}. We see that the shortest vector in $\const{S}{0.25}$ is $u_1 = \pr{\frac{3}{8}, \frac{3}{8}, \frac{1}{4}}$. Similarly, the shortest vector in $\const{S}{0.75}$ is $u_2 = \pr{\frac{1}{8}, \frac{1}{8}, \frac{3}{4}}$. Thus the CM distance is equal to
\[
\dist{D_1}{D_2}{S} = \sqrt{3}\norm{u_1-u_2} = \sqrt{3}\sqrta{\frac{2^2}{8^2}+\frac{2^2}{8^2}+\frac{2^2}{4^2}} = \frac{3}{\sqrt{8}}.
\]
The covariance of $S$ is equal to $\cov{}{}{S} = \frac{1}{3}-\frac{1}{3}\frac{1}{3} = \frac{2}{9}$. Thus Theorem~\ref{thr:calc} gives us
\[
\dist{D_1}{D_2}{S} = \sqrta{\cov{}{-1}{S}\pr{\frac{3}{4}-\frac{1}{4}}^2} = \sqrta{\frac{9}{2}\pr{\frac{2}{4}}^2} = \frac{3}{\sqrt{8}}.
\]
\end{example}
From Theorem~\ref{thr:calc} we see a reason to have the factor $\sqrt{\abs{\Omega}}$ in Eq.~\ref{eq:def}: Assume that we have two data sets $D_1$ and $D_2$ and a feature function $S$. We define a new sample space $\Omega' = \set{\pr{\omega, b} \mid \omega \in \Omega, b = 0, 1}$ and transform the original data sets into new ones by setting $D_i' = \set{\pr{\omega, 0} \mid \omega \in D_i}$. We also expand $S$ into $\Omega'$  by setting $S'(\omega, 1) = S'(\omega, 0) = S(\omega)$. Note that $S(D_i) = S'(D_i')$ and that $\cov{}{}{S} = \cov{}{}{S'}$ so Theorem~\ref{thr:calc} says that the CM distance has not changed during this transformation. This is very reasonable since we did not actually change anything essential: We simply added a bogus variable into the sample space, and we ignored this variable during the feature extraction. The size of the new sample space is $\abs{\Omega'} = 2\abs{\Omega}$. This means that the difference $\norm{u_1-u_2}$ in Eq.~\ref{eq:def} is smaller by the factor $\sqrt{2}$. The factor $\sqrt{\abs{\Omega}}$ is needed to negate this effect.
\subsection{Properties}
\label{sec:properties}
We will now list some important properties of $\dist{D_1}{D_2}{S}$.
\begin{theorem} $\dist{D_1}{D_2}{S}$ is a pseudo metric.
\label{thr:metric}
\end{theorem}

The following theorem says that adding external data set to the original data sets makes the distance smaller which is very reasonable property.

\begin{theorem}
Assume three data sets $D_1$, $D_2$, and $D_3$ over the same set of items. Assume further that $D_1$ and $D_2$ have the same number of data points and let $\epsilon = \frac{\abs{D_3}}{\abs{D_1}+\abs{D_3}}$. Then \[\dist{D_1 \cup D_3}{D_2 \cup D_3}{S} = (1-\epsilon)\dist{D_1}{D_2}{S}.\]
\label{thr:augdata}
\end{theorem}

\begin{theorem}
Let $A$ be a $M \times N$ matrix and $b$ a vector of length $M$. Define $T(\omega) = AS(\omega)+b$. It follows that $\dist{D_1}{D_2}{T} \leq \dist{D_1}{D_2}{S}$ for any $D_1$ and $D_2$.
\label{thr:linear}
\end{theorem}
\begin{corollary}
Adding extra feature functions cannot decrease the distance.
\end{corollary}
\begin{corollary}
Let $A$ be an invertible $N \times N$ matrix and $b$ a vector of length $N$. Define $T(\omega) = AS(\omega)+b$. It follows that $\dist{D_1}{D_2}{T} = \dist{D_1}{D_2}{S}$ for any $D_1$ and $D_2$.
\label{thr:ind}
\end{corollary}

Corollary~\ref{thr:ind} has an interesting interpretation. Note that $\freq{T}{D} = A\freq{S}{D}+b$ and that $\freq{S}{D} = A^{-1}\pr{\freq{T}{D}-b}$. This means that if we know the frequencies $\freq{S}{D}$, then we can infer the frequencies $\freq{T}{D}$ without a new data scan. Similarly, we can infer $\freq{S}{D}$ from $\freq{T}{D}$. We can interpret this relation by thinking that $\freq{S}{D}$ and $\freq{T}{D}$ are merely different representations of the same feature information. Corollary~\ref{thr:ind} says that the CM distance is equal for any such representation.

\subsection{Alternative Characterisation of the CM Distance}
We derived our distance using geometrical interpretation of the distribution sets. In this section we will provide an alternative way for deriving the CM distance. Namely, we will show that if some distance is of Mahalanobis type and satisfies some mild assumptions, then this distance is proportional to the CM distance. The purpose of this theorem is to provide more theoretical evidence to our distance.

We say that a distance $d$ is of Mahalanobis type if
\[
\dista{D_1}{D_2}{S}^2 = \pr{\theta_1-\theta_2}^TC(S)^{-1}\pr{\theta_1-\theta_2},
\]
where $\theta_1 = \freq{S}{D_1}$ and $\theta_2 = \freq{S}{D_2}$ and $C(S)$ maps a feature function $S$ to a symmetric $N\times N$ matrix. Note that if $C(S) = \cov{}{}{S}$, then the distance $d$ is the CM distance. We set $\mathbb{M}$ to be the collection of all distances of Mahalanobis type. Can we justify the decision that we examine only the distances included in $\mathbb{M}$? One reason is that a distance belonging to $\mathbb{M}$ is guaranteed to be a metric. The most important reason, however, is the fact that we can evaluate the distance belonging to $\mathbb{M}$ efficiently (assuming, of course, that we can evaluate $C(S)$).

Let $d \in \mathbb{M}$ and assume that it satisfies two additional assumptions:
\begin{enumerate}
\item If $A$ is an $M \times N$ matrix and $b$ is a vector of length $M$ and if we set $T(\omega) = AS(\omega)+b$, then $C(T) = AC(S)A^T$.
\label{as:1}
\item Fix two points $\omega_1$ and $\omega_2$. Let $\funcdef{\sigma}{\Omega}{\Omega}$ be a function swapping $\omega_1$ and $\omega_2$ and mapping everything else to itself. Define $U(\omega) = S(\sigma(\omega))$. Then $\dista{\sigma(D_1)}{\sigma(D_2)}{U} = \dista{D_1}{D_2}{S}$.
\label{as:2}
\end{enumerate}
The first assumption can be partially justified if we consider that $A$ is an invertible square matrix. In this case the assumption is identical to $\dista{\cdot}{\cdot}{AS+b} = \dista{\cdot}{\cdot}{S}$. This is to say that the distance is independent of the representation of the frequency information. This is similar to Corollary~\ref{thr:ind} given in Section~\ref{sec:properties}. We can construct a distance that would satisfy Assumption~\ref{as:1} in the invertible case but fail in a general case. We consider such distances pathological and exclude them by making a broader assumption. To justify Assumption~\ref{as:2} note that the frequencies have not changed, that is, $\freq{U}{\sigma(D)} = \freq{S}{D}$. Only the representation of single data points have changed. Our argument is that the distance should be based on the frequencies and not on the values of the data points.
\begin{theorem}
Let $d \in M$ satisfying Assumptions~\ref{as:1}~and~\ref{as:2}. If $C(S)$ is invertible, then there is a constant $c>0$, not depending on $S$, such that $\dista{\cdot}{\cdot}{S} = c\dist{\cdot}{\cdot}{S}$.
\label{thr:char}
\end{theorem}

%% file: binary.tex
\section{The CM distance and Binary Data Sets}
\label{sec:cmbin}
In this section we will concentrate on the distances between binary data sets. We will consider the CM distance based on itemset frequencies, a very popular statistics in the literature concerning binary data mining. In the first subsection we will show that a more natural way of representing the CM distance is to use parity frequencies. We also show that we can evaluate the distance in linear time. In the second subsection we will provide more theoretical evidence why the CM distance is a good distance between binary data sets.
\subsection{The CM Distance and Itemsets}
\label{sec:cmitemsets}
We begin this section by giving some definitions. We set the sample space $\Omega$ to be
\[
\Omega = \set{\omega \mid \omega = \enpr{\omega_1}{\omega_K}, \omega_i = 0,1},
\]
that is, $\Omega$ is the set of all binary vectors of length $K$. Note that $\abs{\Omega} = 2^K$. It is custom that each dimension in $\Omega$ is identified with some symbol. We do this by assigning the symbol $a_i$ to the $i^\text{th}$ dimension. These symbols are called \emph{attributes} or \emph{items}. Thus when we speak of the attribute $a_i$ we refer to the $i^\text{th}$ dimension. We denote the set of all items by $\mathbb{A} = \enset{a_1}{a_K}$. A non-empty subset of $\mathbb{A}$ is called \emph{itemset}. 

A \emph{boolean formula} $\funcdef{S}{\Omega}{\set{0,1}}$ is a feature function mapping a binary vector to a binary value. We are interested in two particular boolean formulae: Assume that we are given an itemset $B = \enset{a_{i_1}}{a_{i_L}}$. We define a \emph{conjunction function} $S_B$ to be
\[
S_B(\omega) = \omega_{i_1} \land \omega_{i_2} \land \cdots \land \omega_{i_K},
\]
that is, $S_B$ results $1$ if and only if all the variables corresponding the itemset $B$ are on. Given a data set $D$ the frequency $S_B(D)$ is called the frequency of the itemset $B$. Conjuction functions are popular and there are a lot of studies in the literature concerning finding itemsets that have large frequency~\citep[see e.g.,][]{agrawal93mining,hand02principles}. We also define a \emph{parity function} $T_B$ to be
\[
T_B(\omega) = \omega_{i_1} \oplus \omega_{i_2} \oplus \cdots \oplus \omega_{i_K},
\]
where $\oplus$ is the binary operator XOR. The function $T_B$ results $1$ if and only if the number of active variables included in $B$ are odd.

A collection $\iset{F}$ of itemsets is said to be \emph{antimonotonic} or \emph{downwardly closed} if each non-empty subset of an itemset included in $\iset{F}$ is also included in $\iset{F}$. Given a collection of itemsets $\iset{F} = \enset{B_1}{B_N}$ we extend our definition of the conjuction function by setting $S_\iset{F} = \envec{S_{B_1}}{S_{B_N}}$. We also define $T_\iset{F} = \envec{T_{B_1}}{T_{B_N}}$.

Assume that we are given an antimonotonic family $\iset{F}$ of itemsets. We can show that there is an invertible matrix $A$ such that $T_\iset{F} = AS_\iset{F}$. In other words, we can get the parity function $T_\iset{F}$ from the conjunction function $T_\iset{F}$ by an invertible linear transformation. Corollary~\ref{thr:ind} now implies that
\begin{equation}
\label{eq:andvsparity}
\dist{D_1}{D_2}{S_\iset{F}} = \dist{D_1}{D_2}{T_\iset{F}},
\end{equation}
for any $D_1$ and $D_2$. The following lemma shows that the covariance matrix $\cov{}{}{T_\iset{F}}$ of the parity function is very simple.
\begin{lemma}
\label{lem:paritycov}
Let $T_\iset{F}$ be a parity function for a family of itemsets $\iset{F}$, then $\cov{}{}{T_\iset{F}} = 0.5I$, that is, the covariance matrix is a diagonal matrix having $0.5$ at the diagonal.
\end{lemma}
Theorem~\ref{thr:calc}, Lemma~\ref{lem:paritycov}, and Eq.~\ref{eq:andvsparity} imply that
\begin{equation}
\dist{D_1}{D_2}{S_\iset{F}} = \sqrt{2}\norm{\theta_1 - \theta_2},
\label{eq:fastev}
\end{equation}
where $\theta_1 = \freq{T_\iset{F}}{D_1}$ and $\theta_2 = \freq{T_\iset{F}}{D_2}$. This identity says that the CM distance can be calculated in $O(N)$ time (assuming that we know the frequencies $\theta_1$ and $\theta_2$). This is better than $O(N^3)$ time implied by Theorem~\ref{thr:calc}.
\begin{example}
\label{ex:ind}
Let $\iset{I} = \set{\set{a_j}\mid j = 1 \ldots K}$ be a family of itemsets having only one item. Note that $T_{\set{a_j}} = S_{\set{a_j}}$. Eq.~\ref{eq:fastev} implies that
\[
\dist{D_1}{D_2}{S_\iset{I}} = \sqrt{2}\norm{\theta_{1}-\theta_{2}},
\]
where $\theta_1$ and $\theta_2$ consists of the marginal frequencies of each $a_j$ calculated from $D_1$ and $D_2$, respectively. In this case the CM distance is simply the $L_2$ distance between the marginal frequencies of the individual attributes. The frequencies $\theta_1$ and $\theta_2$ resemble term frequencies (TF) used in text mining~\citep[see e.g.,][]{baldi03internet}.
\end{example}

\begin{example}
\label{ex:cov}
We consider now a case with a larger set of features. Our motivation for this is that using only the feature functions $S_\iset{I}$ is sometimes inadequate. For example, consider data sets with two items having the same individual frequencies but different correlations. In this case the data sets may look very different but according to our distance they are equal.

Let $\iset{C} = \iset{I} \cup \set{a_ja_k \mid j, k = 1\ldots K, j < k}$ be a family of itemsets such that each set contains at most two items. The corresponding frequencies contain the individual means and the pairwise correlation for all items. Let $S_{a_ja_k}$ be the conjunction function for the itemset $a_ja_k$. Let $\gamma_{jk} = \freq{S_{a_ja_k}}{D_1}-\freq{S_{a_ja_k}}{D_2}$ be the difference between the correlation frequencies. Also, let $\gamma_j = \freq{S_{a_j}}{D_1}-\freq{S_{a_j}}{D_2}$. Since
\[
T_{a_ja_k} = S_{a_j}+S_{a_k}-2S_{a_ja_k}
\]
it follows from Eq.~\ref{eq:fastev} that
\begin{equation}
\dist{D_1}{D_2}{S_\iset{C}}^2 = 2\sum_{j < k}\pr{\gamma_j+\gamma_k-2\gamma_{jk}}^2+2\sum_{j=1}^{K}\gamma_j^2.
\label{eq:distcovform}
\end{equation}
\end{example}

\subsection{Characterisation of the CM Distance for Itemsets}
The identity given in Eq.~\ref{eq:fastev} is somewhat surprising and seems less intuitive. A question arises: why this distance is more natural than some other, say, a simple $L_2$ distance between the itemset frequencies. Certainly, parity functions are less intuitive than conjunction functions. One answer is that the parity frequencies are decorrelated version of the traditional itemset frequencies.

However, we can clarify this situation from another point of view: Let $\iset{A}$ be the set of all itemsets. Assume that we are given two data sets $D_1$ and $D_2$ and define \emph{empirical distributions} $p_1$ and $p_2$ by setting
\[
p_i(\omega) = \frac{\text{number of samples in $D_i$ equal to $\omega$}}{\abs{D_i}}.
\]
The constrained spaces of $S_\iset{A}$ are singular points containing only $p_i$, that is, $\const{S_\iset{A}}{D_i} = \set{p_i}$. This implies that
\begin{equation}
\dist{D_1}{D_2}{S_\iset{A}} = \sqrt{2^K}\norm{p_1-p_2}.
\label{eq:fullinfo}
\end{equation}
In other words, the CM distance is proportional to the $L_2$ distance between the empirical distributions. This identity seems very reasonable. At least, it is more natural than, say, taking $L_2$ distance between the traditional itemset frequencies.

The identity in Eq.~\ref{eq:fullinfo} holds only when we use the features $S_\iset{A}$. However, we can prove that a distance of the Mahalanobis type satisfying the identity in Eq.~\ref{eq:fullinfo} and some additional conditions is equal to the CM distance. Let us explain this in more detail. We assume that we have a distance $d$ having the form
\[
\dista{D_1}{D_2}{S_\iset{F}}^2 = \pr{\theta_1-\theta_2}^TC(S_\iset{F})^{-1}\pr{\theta_1-\theta_2},
\]
where $\theta_1 = \freq{S_\iset{F}}{D_1}$ and $\theta_2 = \freq{S_\iset{F}}{D_2}$ and $C(S_\iset{F})$ maps a conjuction function $S_\iset{F}$ to a symmetric $N\times N$ matrix. The distance $d$ should satisfy the following mild assumptions.
\begin{enumerate}
\item Assume two antimonotonic families of itemsets $\iset{F}$ and $\iset{H}$ such that $\iset{F} \subset \iset{H}$. It follows that $\dista{\cdot}{\cdot}{S_\iset{F}} \leq \dista{\cdot}{\cdot}{S_\iset{H}}$.
\label{as2:2}
\item Adding extra dimensions (but not changing the features) does not change the distance.
\label{as2:3}
\end{enumerate}
The following theorem says that the assumptions and the identity in Eq.~\ref{eq:fullinfo} are sufficient to prove that $d$ is actually the CM distance.
\begin{theorem}
\label{thr:generic} Assume that a Mahalanobis distance $d$ satisfies Assumptions~\ref{as2:2}~and~\ref{as2:3}. Assume also that there is a constant $c_1$ such that
\[
\dista{D_1}{D_2}{S_\iset{A}} = c_1\norm{p_1-p_2}.
\]
Then it follows that for any antimonotonic family $\iset{F}$ we have
\[
\dista{D_1}{D_2}{S_\iset{F}} = c_2\dist{D_1}{D_2}{S_\iset{F}},
\]
for some constant $c_2$.
\end{theorem}

%% file: sequence.tex
\section{The CM distance and Event Sequences}
\label{sec:sequences}
In the previous section we discussed about the CM distance between the binary data sets. We will use similar approach to define the CM distance between sequences.

An \emph{event sequence} $s$ is a finite sequence whose symbols belong to a finite alphabet $\Sigma$. We denote the length of the event sequence $s$ by $\abs{s}$, and by $s(i, j)$ we mean a subsequence starting from $i$ and ending at $j$. The subsequence $s(i, j)$ is also known as \emph{window}. A popular choice for statistics of event sequences are \emph{episodes}~\citep{hand02principles}. A \emph{parallel episode} is represented by a subset of the alphabet $\Sigma$. A window of $s$ satisfies a parallel episode if all the symbols given in the episode occur in the window. Assume that we are given an integer $k$. Let $W$ be a collection of windows of $s$ having the length $k$. A \emph{frequency} of a parallel episode is the proportion of windows in $W$ satisfying the episode. We should point out that this mapping destroys the exact ordering of the sequence. On the other hand, if some symbols occur often close to each other, then the episode consisting of these symbols will have a high frequency.

In order to apply the CM distance we will now describe how we can transform a sequence $s$ to a binary data set. Assume that we are given a window length $k$. We transform a window of length $k$ into a binary vector of length $\abs{\Sigma}$ by setting $1$ if the corresponding symbol occurs in the window, and $0$ otherwise. Let $D$ be the collection of these binary vectors. We have now transformed the sequence $s$ to the binary data set $D$. Note that parallel episodes of $s$ are represented by itemsets of $D$.

This transformation enables us to use the CM distance. Assume that we are given two sequences $s_1$ and $s_2$, a collection of parallel episodes $\iset{F}$, and a window length $k$. First, we transform the sequences into data sets $D_1$ and $D_2$. We set the CM distance between the sequences $s_1$ and $s_2$ to be $\dist{D_1}{D_2}{S_\iset{F}}$.

%% file: feature.tex
\section{Feature Selection}
\label{sec:feature}
We will now discuss briefly about feature selection --- a subject that we have taken for granted so far. The CM distance depends on a feature function $S$. How can we choose a good set of features?

Assume for simplicity that we are dealing with binary data sets. Eq.~\ref{eq:fullinfo} tells us that if we use all itemsets, then the CM distance is $L_2$ distance between empirical distributions. However, to get a reliable empirical distribution we need an exponential number of data points. Hence we can use only some subset of itemsets as features. The first approach is to make an expert choice without seeing data. For example, we could decide that the feature function is $S_\iset{I}$, the means of the individual attributes, or $S_\iset{C}$, the means of individual attributes and the pairwise correlation.

The other approach is to infer a feature function from the data sets. At first glimpse this seems an application of feature selection. However, traditional feature selection fails: Let $S_\iset{I}$ be the feature function representing the means of the individual attributes and let $S_\iset{A}$ be the feature function containing all itemsets. Let $\omega$ be a binary vector. Note that if we know $S_\iset{I}(\omega)$, then we can deduce $S_\iset{A}(\omega)$. This means that $S_\iset{I}$ is a \emph{Markov blanket}~\citep{pearl88reasoning} for $S_\iset{A}$. Hence we cannot use the Markov blanket approach to select features. The essential part is that the traditional feature selection algorithms deal with the \emph{individual} points. We try to select features for whole data sets.

Note that feature selection algorithms for singular points are based on training data, that is, we have data points divided into clusters. In other words, when we are making traditional feature selection we \emph{know} which points are close and which are not. In order to make the same ideas work for data sets we need to have similar information, that is, we need to know which data sets are close to each other, and which are not. Such an information is rarely provided and hence we are forced to seek some other approach.

We suggest a simple approach for selecting itemsets by assuming that frequently occurring itemsets are interesting. Assume that we are given a collection of data sets $D_i$ and a threshold $\sigma$. Let $\iset{I}$ be the itemsets of order one. We define $\iset{F}$ such that $B \in \iset{F}$ if and only if $B \in \iset{I}$ or that $B$ is a $\sigma$-frequent itemset for some $D_i$.

%% file: related.tex
\section{Related Work}
\label{sec:related}
In this section we discuss some existing methods for comparing data sets and compare the evaluation algorithms. The execution times are summarised in Table~\ref{tab:times}.
\begin{table}[ht!]
\centering
\begin{tabular}{rr}
\toprule
Distance & Time \\
\midrule
CM distance (general case) & $O(NM+N^2\abs{\Omega}+N^3)$\\
CM distance (known cov. matrix) & $O(NM+N^3)$\\
CM distance (binary case) & $O(NM+N)$ \\
Set distances & $O(M^3)$\\
Kullback-Leibler &  $O(NM+N\abs{\Omega})$\\
Fischer's Information &  $O(NM+N^2\abs{D_2}+N^3)$\\
\bottomrule
\end{tabular}
\caption{Comparison of the execution times of various distances. The number $M = \abs{D_1}+\abs{D_2}$ is the number of data points in total. The $O(NM)$ term refers to the time needed to evaluate the frequencies $\freq{S}{D_1}$ and $\freq{S}{D_2}$. Kullback-Leibler distance is solved using Iterative Scaling algorithm in which one round has $N$ steps and one step is executed in $O(\abs{\Omega})$ time.}
\label{tab:times}
\end{table}
\subsection{Set Distances}
One approach to define a data set distance is to use some natural distance between single data points and apply some known set distance. \citet{eiter97distance} show that some data set distances defined in this way can be evaluated in cubic time. However, this is too slow for us since we may have a vast amount of data points. The other downsides are that these distances may not take into account the statistical nature of data which may lead into problems.
\subsection{Edit Distances}
We discuss in Section~\ref{sec:sequences} of using the CM distance for event sequences. Traditionally, edit distances are used for comparing event sequences. The most famous edit distance is Levenshtein distance~\citep{levenshtein66distance}. However, edit distances do not take into account the statistical nature of data. For example, assume that we have two sequences generated such that the events are sampled from the uniform distribution independently of the previous event (a zero-order Markov chain). In this case the CM distance is close to $0$ whereas the edit distance may be large. Roughly put, the CM distance measures the dissimilarity between the statistical characteristics whereas the edit distances operate at the symbol level.

\subsection{Minimum Discrimination Approach}
There are many distances for distributions~\citep[see][for a nice review]{baseville89distance}. From these distances the CM distance resembles the statistical tests involved with Minimum Discrimination Theorem~\citep[see][]{kullback68information, csiszar75divergence}. In this framework we are given a feature function $S$ and two data sets $D_1$ and $D_2$. From the set of distributions $\constp{S}{D_i}$ we select a distribution maximising the entropy and denote it by $p^{ME}_i$. The distance itself is the Kullback-Leibler divergence between $p^{ME}_1$ and $p^{ME}_2$. It has been empirically shown that $p^{ME}_i$ represents well the distribution from which $D_i$ is generated~\citep[see][]{mannila99prediction}. The downsides are that this distance is not a metric (it is not even symmetric), and that the evaluation time of the distance is infeasible: Solving $p^{ME}_i$ is \textbf{NP}-hard~\citep{cooper90complexity}. We can approximate the Kullback-Leibler distance by Fischer's information, that is,
\[
\kl{p^{ME}_1}{p^{ME}_2} \approx \frac{1}{2}\pr{\theta_1 - \theta_2}^T\cova{p^{ME}_2}{-1}{S}\pr{\theta_1 - \theta_2},
\]
where $\theta_i = \freq{S}{D_i}$ and $\cova{p^{ME}_2}{}{S}$ is the covariance matrix of $S$ taken with respect to $p^{ME}_2$~\citep[see][]{kullback68information}. This resembles greatly the equation in Theorem~\ref{thr:calc}. However, in this case the covariance matrix depends on data sets and thus generally this approximation is not a metric. In addition, we do not know $p^{ME}_2$ and hence we cannot evaluate the covariance matrix. We can, however, estimate the covariance matrix from $D_2$, that is,
\[
\cova{p^{ME}_2}{}{S} \approx \frac{1}{\abs{D_2}}\sum_{\omega \in D_2}S(\omega)S(\omega)^T -
\frac{1}{\abs{D_2}^2}\spr{\sum_{\omega \in D_2}S(\omega)}\spr{\sum_{\omega \in D_2}S(\omega)^T}.
\]
The execution time of this operation is $O(N^2\abs{D_2})$.

%% file: tests.tex
\section{Empirical Tests}
\label{sec:tests}
In this section we describe our experiments with the CM distance. We begin by examining the effect of different feature functions. We continue studying the distance by applying clustering algorithms, and finally we represent some interpretations to the results.

In many experiments we use a base distance $d_{U}$ defined as the $L_2$ distance between the itemset frequencies, that is,
\begin{equation}
\distu{D_1}{D_2}{S} = \sqrt{2}\norm{\theta_1 - \theta_2},
\label{eq:base}
\end{equation}
where $\theta_i$ are the itemset frequencies $\theta_i = \freq{S}{D_i}$. This type of distance was used in~\cite{hollmen03mixture}. Note that $\distu{D_1}{D_2}{ind} = \dist{D_1}{D_2}{ind}$, where $ind$ is the feature set containing only individual means.

\subsection{Real World Data Sets}
We examined the CM distance with several real world data sets and several feature sets. We had $7$ data sets: \dtname{Bible}, a collection of $73$ books from the Bible\footnote{The books were taken from \url{http://www.gutenberg.org/etext/8300} in 20. July 2005}, \dtname{Addresses}, a collection of $55$ inaugural addresses given by the presidents of the U.S.\footnote{The addresses were taken from \url{http://www.bartleby.com/124/} in 17. August 2005}, \dtname{Beatles}, a set of lyrics from $13$ studio albums made by the Beatles, \dtname{20Newsgroups}, a collection of 20 newsgroups\footnote{The data set was taken from \url{http://www.ai.mit.edu/~jrennie/20Newsgroups/}, a site hosted by Jason Rennie, in 8. June, 2001.}, \dtname{TopGenres}, plot summaries for top rated movies of 8 different genres, and \dtname{TopDecades}, plot summaries for top rated movies from 8 different decades\footnote{The movie data sets were taken from \url{http://www.imdb.com/Top/} in 1. January, 2006}. \dtname{20Newsgroups} contained (in that order) 3 religion groups, 3 of politics, 5 of computers, 4 of science, 4 recreational, and \dtname{misc.forsale}. \dtname{TopGenres} consisted (in that order) of \dtname{Action}, \dtname{Adventure}, \dtname{SciFi}, \dtname{Drama}, \dtname{Crime}, \dtname{Horror}, \dtname{Comedy}, and \dtname{Romance}. The decades for \dtname{TopDecades} were 1930--2000. Our final data set, \dtname{Abstract}, was composed of abstracts describing NSF awards from 1990--1999\footnote{The data set was taken from \url{http://kdd.ics.uci.edu/databases/nsfabs/nsfawards.data.html} in 13. January, 2006}.

\dtname{Bible} and \dtname{Addresses} were converted into binary data sets by taking subwindows of length $6$ (see the discussion in Section~\ref{sec:sequences}). We reduced the number of attributes to $1000$ by using the mutual information gain. \dtname{Beatles} was preprocessed differently: We transformed each song to its binary bag-of-words representation and selected $100$ most informative words. In \dtname{20Newsgroups} a transaction was a binary bag-of-words representation of a single article. Similarly, In \dtname{TopGenres} and in \dtname{TopDecades} a transaction corresponded to a single plot summary. We reduced the number of attributes in these three data sets to $200$ by using the mutual information gain. In \dtname{Abstract} a data set represented one year and a transaction was a bag-of-words representation of a single abstract. We reduced the dimension of \dtname{Abstract} to $1000$.

\subsection{The Effect of Different Feature Functions}
We begin our experiments by studying how the CM distance (and the base distance) changes as we change features.

We used $3$ different sets of features: \ftname{ind}, the independent means, \ftname{cov}, the independent means along with the pairwise correlation, and \ftname{freq}, a family of frequent itemsets obtained by using \alname{APriori}~\citep{agrawal96apriori}. We adjusted the threshold so that \ftname{freq} contained $10K$ itemsets, where $K$ is the number of attributes.

%We had noticed in the previous section that the distance $\dist{\cdot}{\cdot}{S_\iset{C}}$ resembles the distance $\dist{\cdot}{\cdot}{S_\iset{I}}$. In this section we will examine this behaviour in more details. That is, we will study the distances $\dist{\cdot}{\cdot}{S_\iset{C}}$ and $\dist{\cdot}{\cdot}{S_\sigma}$ as functions of $\dist{\cdot}{\cdot}{S_\iset{I}}$.

We plotted the CM distances and the base distances as functions of $\dcm{ind}$. The results are shown in Figure~\ref{fig:scatterplots}. Since the number of constraints varies, we normalised the distances by dividing them with $\sqrt{N}$, where $N$ is the number of constraints. In addition, we computed the correlation of each pair of distances. These correlations are shown in Table~\ref{tab:corrs}.

\begin{figure}[htb!]
\center
\includegraphics[width=6cm]{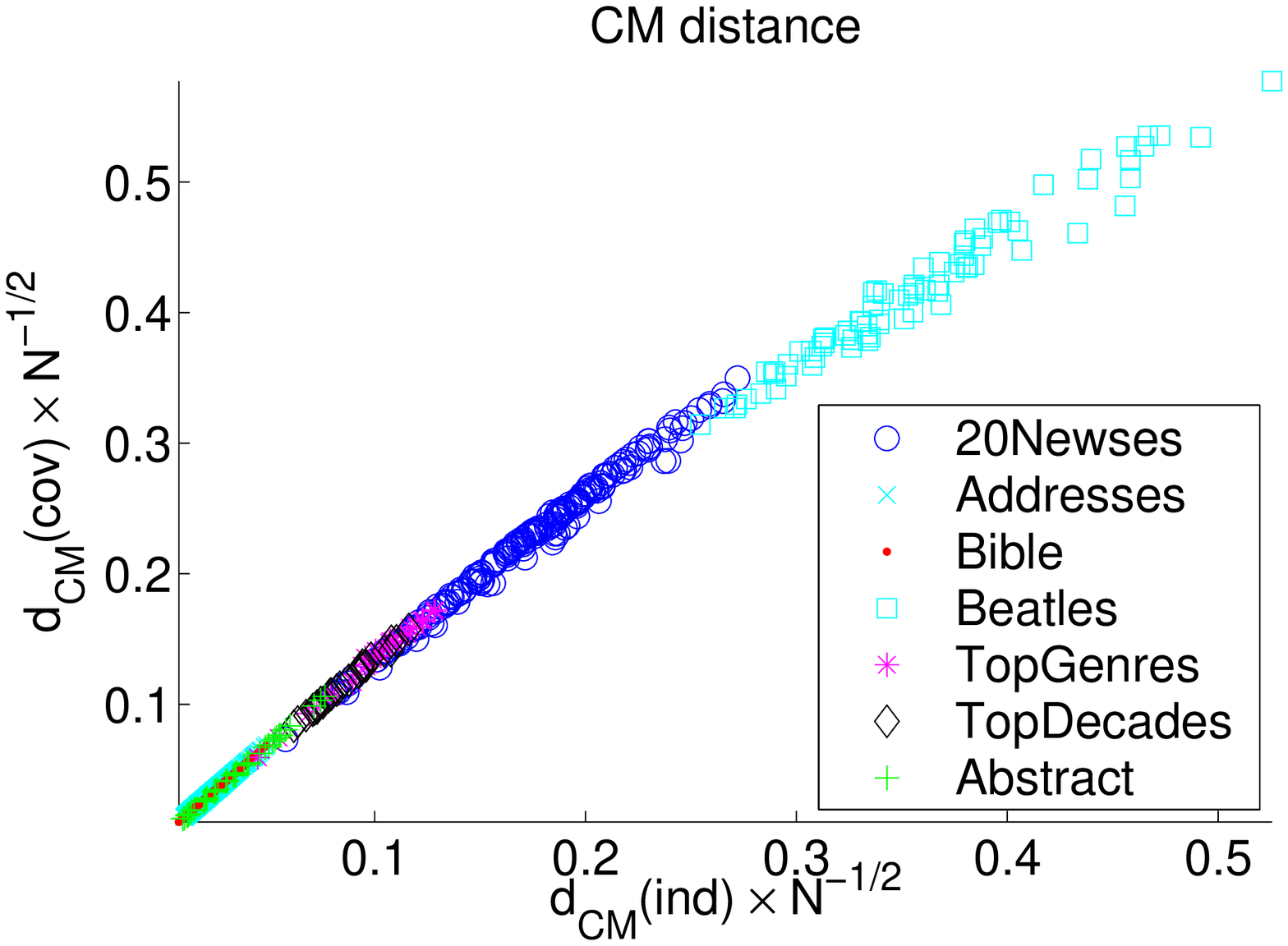}
\includegraphics[width=6cm]{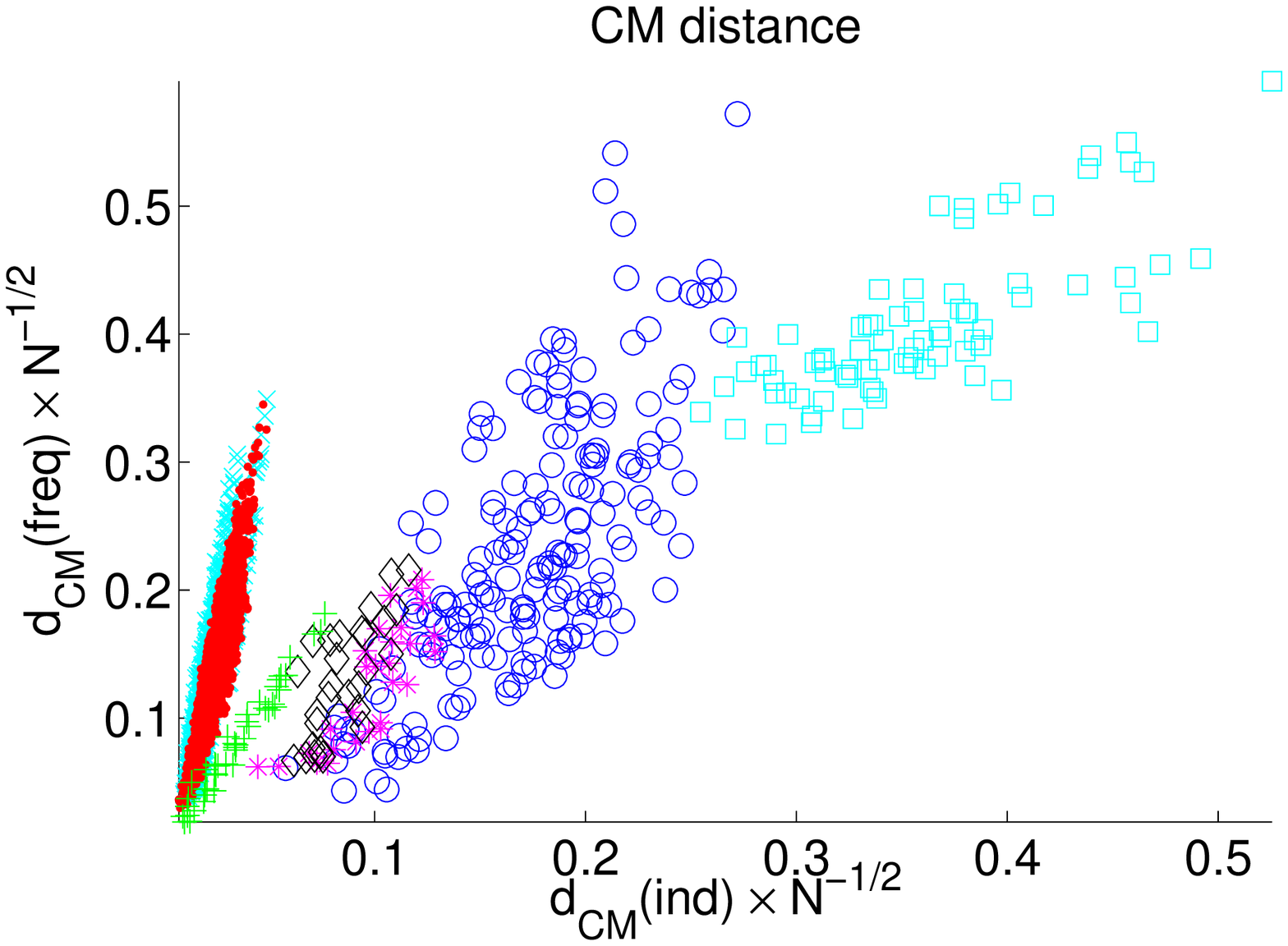}
\includegraphics[width=6cm]{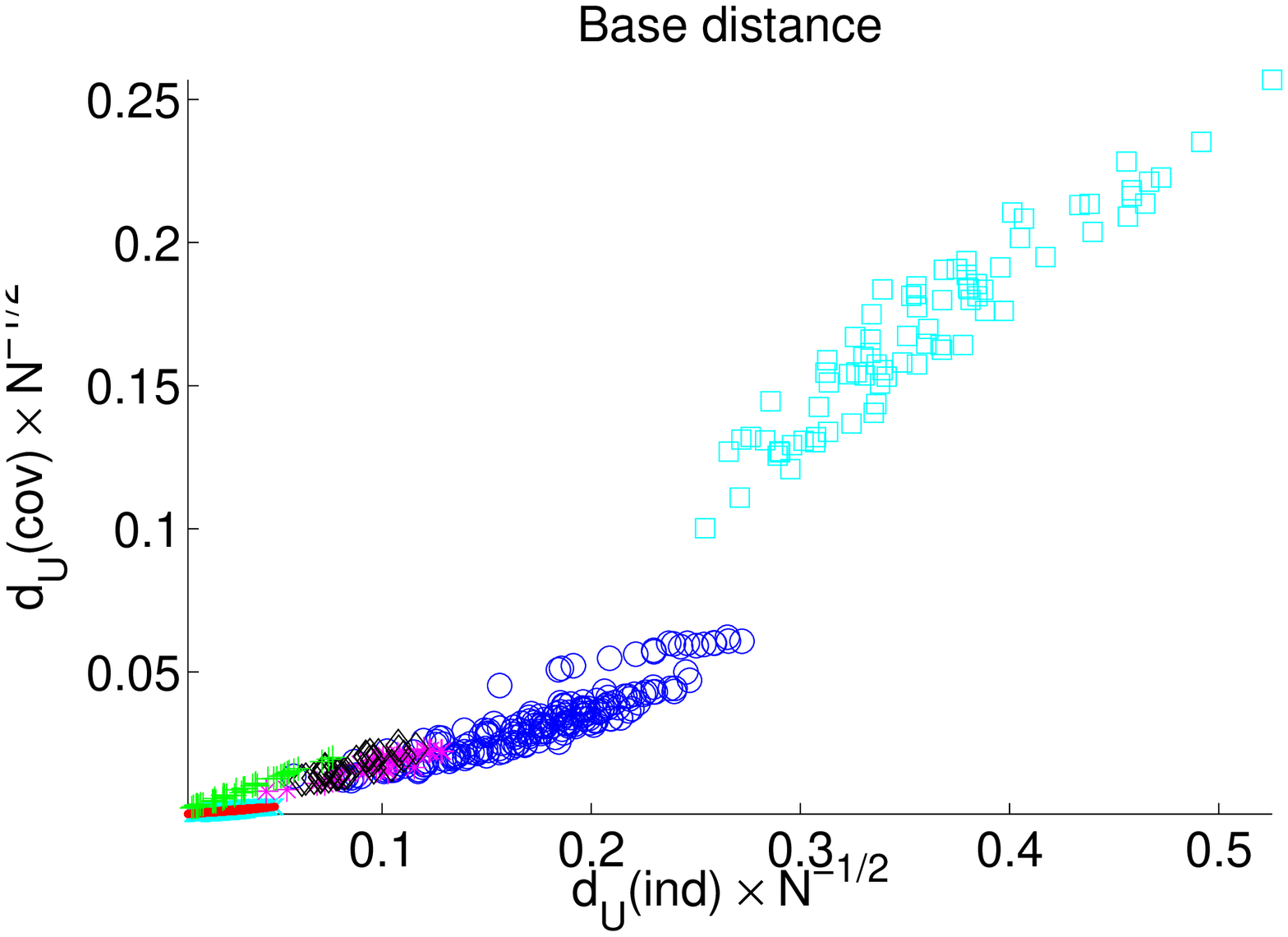}
\includegraphics[width=6cm]{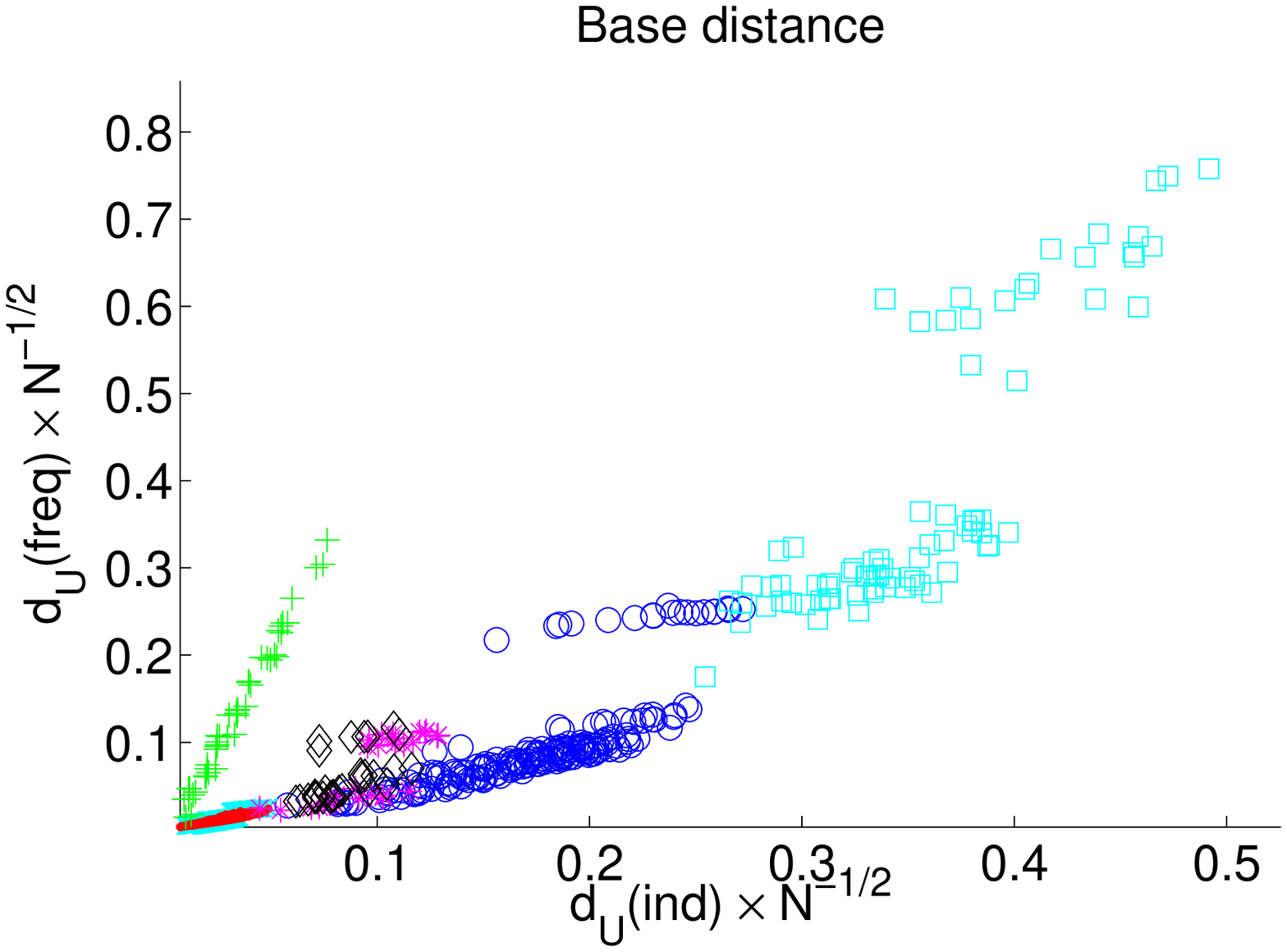}
\caption{CM and base distances as functions of $\dcm{ind}$. A point represents a distance between two data sets. The upper two figures contain the CM distances while the lower two contain the base distance. The distances were normalised by dividing $\sqrt{N}$, where $N$ is the number of constraints. The corresponding correlations are given in Table~\ref{tab:corrs}. Note that $x$-axis in the left (right) two figures are equivalent.}
\label{fig:scatterplots}
\end{figure}

\begin{table}[htb!]
\centering
\begin{tabular}{rrrrrrrrr}
\toprule
& \multicolumn{3}{c}{$d_{CM}$ vs. $d_{CM}$} & \multicolumn{3}{c}{$d_U$ vs. $d_U$} & \multicolumn{2}{c}{$d_{CM}$ vs. $d_U$} \\
\cmidrule{2-9}
          & \ftname{cov} & \ftname{freq} & \ftname{freq} & \ftname{cov} & \ftname{freq} & \ftname{freq} & \ftname{cov} & \ftname{freq} \\
Data set  & \ftname{ind} & \ftname{ind} & \ftname{cov} & \ftname{ind} & \ftname{ind} & \ftname{cov} & \ftname{cov} & \ftname{freq} \\
\midrule
\dtname{20Newsgroups} & $0.996$ & $0.725$ & $0.733$ & $0.902$ & $0.760$ & $0.941$ & $0.874$ & $0.571$ \\
\dtname{Addresses} & $1.000$ & $0.897$ & $0.897$ & $0.974$ & $0.927$ & $0.982$ & $0.974$ & $0.743$ \\
\dtname{Bible} & $1.000$ & $0.895$ & $0.895$ & $0.978$ & $0.946$ & $0.989$ & $0.978$ & $0.802$ \\
\dtname{Beatles} & $0.982$ & $0.764$ & $0.780$ & $0.951$ & $0.858$ & $0.855$ & $0.920$ & $0.827$ \\
\dtname{TopGenres} & $0.996$ & $0.817$ & $0.833$ & $0.916$ & $0.776$ & $0.934$ & $0.927$ & $0.931$ \\
\dtname{TopDecades} & $0.998$ & $0.735$ & $0.744$ & $0.897$ & $0.551$ & $0.682$ & $0.895$ & $0.346$ \\
\dtname{Abstract} & $1.000$ & $0.985$ & $0.985$ & $0.996$ & $0.993$ & $0.995$ & $0.996$ & $0.994$ \\
Total & $0.998$ & $0.702$ & $0.709$ & $0.934$ & $0.894$ & $0.938$ & $0.910$ & $0.607$ \\
\bottomrule
\end{tabular}
\caption{Correlations for various pairs of distances. A column represents a pair of distances and a row represents a single data set. For example, the correlation between $\dcm{ind}$ and $\dcm{cov}$ in \dtname{20Newsgroups} is $0.996$. The last row is the correlation obtained by using the distances from all data sets simultaneously. Scatterplots for the columns 1--2 and 4--5 are given in Fig.~\ref{fig:scatterplots}.}
\label{tab:corrs}
\end{table}

Our first observation from the results is that $\dcm{cov}$ resembles $\dcm{ind}$ whereas $\dcm{freq}$ produces somewhat different results.

The correlations between $\dcm{cov}$ and $\dcm{ind}$ are stronger than the correlations between $\du{cov}$ and $\du{ind}$. This can be explained by examining Eq.~\ref{eq:distcovform} in Example~\ref{ex:cov}. If the dimension is $K$, then the itemsets of size $1$, according to Eq.~\ref{eq:distcovform}, involve $\frac{1}{2}K(K - 1) + K$ times in computing $\dcm{cov}$, whereas in computing $\du{cov}$ they involve only $K$ times. Hence, the itemsets of size $2$ have smaller impact in $\dcm{cov}$ than in $\du{cov}$.

On the other hand, the correlations between $\dcm{freq}$ and $\dcm{ind}$ are weaker than the correlations between $\du{freq}$ and $\du{ind}$, implying that the itemsets of higher order have stronger impact on the CM distance.

\subsection{Clustering Experiments}
In this section we continue our experiments by applying clustering algorithms to the distances. Our goal is to compare the clusterings obtained from the CM distance to those obtained from the base distance (given in Eq.~\ref{eq:base}).

We used $3$ different clustering algorithms: a hierarchical clustering with complete linkage, a standard K-median, and a spectral algorithm by~\cite{ng02clustering}. Since each algorithm takes a number of clusters as an input parameter, we varied the number of clusters between $3$ and $5$. We applied clustering algorithms to the distances $\dcm{cov}$, $\dcm{freq}$, $\du{cov}$, and $\du{freq}$, and compare the clusterings obtained from $\dcm{cov}$ against the clusterings obtained from $\du{cov}$, and similarly compare the clusterings obtained from $\dcm{freq}$ against the clusterings obtained from $\du{freq}$.

We measured the performance using $3$ different clustering indices: a ratio $r$ of the mean of the intra-cluster distances and the mean of the inter-cluster distances, Davies-Bouldin (DB) index~\citep{davies79index}, and Calinski-Harabasz (CH) index~\citep{calinski74index}.

The obtained results were studied in the following way: Given a data set and a performance index, we calculated the number of algorithms in which $\dcm{cov}$ outperformed $\du{cov}$. The distances $\dcm{freq}$ and $\du{freq}$ were handled similarly. The results are given in Table~\ref{tab:clustperf2}. We also calculated the number of data sets in which $\dcm{cov}$ outperformed $\du{cov}$, given an algorithm and an index. These results are given in Table~\ref{tab:clustperf1}.

\begin{table}[htb!]
\centering
\begin{tabular}{rrrrrrrrrr}
\toprule
& & \multicolumn{3}{c}{$\dcm{cov}$ vs. $\du{cov}$} & \multicolumn{3}{c}{$\dcm{freq}$ vs. $\du{freq}$} \\
\cmidrule{3-8}
& Data set  &  $r$ & $DB$ & $CH$ & $r$ & $DB$ & $CH$ & Total & $P$ \\
\midrule
1. & \dtname{20Newsgroups} & $0/9$ & $2/9$ & $7/9$ & $8/9$ & $5/9$ & $9/9$ & $31/54$ & $0.22$ \\
2. & \dtname{Speeches} & $9/9$ & $6/9$ & $3/9$ & $9/9$ & $9/9$ & $9/9$ & $\mathbf{45/54}$ & $0.00$ \\
3. & \dtname{Bible} & $9/9$ & $7/9$ & $2/9$ & $9/9$ & $7/9$ & $9/9$ & $\mathbf{43/54}$ & $0.00$ \\
4. & \dtname{Beatles} & $0/9$ & $3/9$ & $6/9$ & $0/9$ & $1/9$ & $0/9$ & $\mathbf{10/54}$ & $0.00$ \\
5. & \dtname{TopGenres} & $0/9$ & $4/9$ & $5/9$ & $0/9$ & $1/9$ & $0/9$ & $\mathbf{10/54}$ & $0.00$ \\
6. & \dtname{TopDecades} & $3/9$ & $7/9$ & $2/9$ & $7/9$ & $7/9$ & $9/9$ & $\mathbf{35/54}$ & $0.02$ \\
7. & \dtname{Abstract} & $9/9$ & $8/9$ & $1/9$ & $0/9$ & $2/9$ & $1/9$ & $21/54$ & $0.08$ \\
\midrule
& Total & $30/63$ & $37/63$ & $26/63$ & $33/63$ & $32/63$ & $37/63$ & $195/378$ & $0.50$ \\
& $P$ & $0.61$ & $0.13$ & $0.13$ & $0.61$ & $0.80$ & $0.13$\\
\bottomrule
\end{tabular}
\caption{Summary of the performance results of the CM distance versus the base distance. A single entry contains the number of clustering algorithm configurations (see Column $1$ in Table~\ref{tab:clustperf1}) in which the CM distance was better than the base distance. The $P$-value is the standard Fisher's sign test.}
\label{tab:clustperf2}
\end{table}

\begin{table}[htb!]
\centering
\begin{tabular}{rrrrrrrrrr}
\toprule
& & \multicolumn{3}{c}{$\dcm{cov}$ vs. $\du{cov}$} & \multicolumn{3}{c}{$\dcm{freq}$ vs. $\du{freq}$} \\
\cmidrule{3-8}
& Algorithm  & $r$ & $DB$ & $CH$ & $r$ & $DB$ & $CH$ & Total & $P$ \\
\midrule
1. & \alname{K-med(3)} & $4/7$ & $2/7$ & $5/7$ & $4/7$ & $4/7$ & $4/7$ & $23/42$ & $0.44$ \\
2. & \alname{K-med(4)} & $4/7$ & $4/7$ & $3/7$ & $4/7$ & $4/7$ & $4/7$ & $23/42$ & $0.44$ \\
3. & \alname{K-med(5)} & $4/7$ & $4/7$ & $3/7$ & $4/7$ & $4/7$ & $4/7$ & $23/42$ & $0.44$ \\
4. & \alname{link(3)} & $3/7$ & $4/7$ & $3/7$ & $2/7$ & $3/7$ & $4/7$ & $19/42$ & $0.44$ \\
5. & \alname{link(4)} & $3/7$ & $4/7$ & $3/7$ & $4/7$ & $3/7$ & $4/7$ & $21/42$ & $0.88$ \\
6. & \alname{link(5)} & $3/7$ & $3/7$ & $4/7$ & $4/7$ & $2/7$ & $4/7$ & $20/42$ & $0.64$ \\
7. & \alname{spect(3)} & $3/7$ & $6/7$ & $1/7$ & $3/7$ & $4/7$ & $4/7$ & $21/42$ & $0.88$ \\
8. & \alname{spect(4)} & $3/7$ & $4/7$ & $3/7$ & $4/7$ & $4/7$ & $4/7$ & $22/42$ & $0.64$ \\
9. & \alname{spect(5)} & $3/7$ & $6/7$ & $1/7$ & $4/7$ & $4/7$ & $5/7$ & $23/42$ & $0.44$ \\
\midrule
& Total & $30/63$ & $37/63$ & $26/63$ & $33/63$ & $32/63$ & $37/63$ & $195/378$ & $0.50$ \\
& $P$ & $0.61$ & $0.13$ & $0.13$ & $0.61$ & $0.80$ & $0.13$\\
\bottomrule
\end{tabular}
\caption{Summary of the performance results of the CM distance versus the base distance. A single entry contains the number of data sets (see Column $1$ in Table~\ref{tab:clustperf2}) in which the CM distance was better than the base distance. The $P$-value is the standard Fisher's sign test.}
\label{tab:clustperf1}
\end{table}

We see from Table~\ref{tab:clustperf2} that the performance of CM distance against the base distance depends on the data set. For example, the CM distance has tighter clusterings in \dtname{Speeches}, \dtname{Bible}, and \dtname{TopDecade} whereas the base distance outperforms the CM distance in \dtname{Beatles} and \dtname{TopGenre}.

Table~\ref{tab:clustperf1} suggests that the overall performance of the CM distance is as good as the base distance. The CM distance obtains a better index $195$ times out of $378$. The statistical test suggests that this is a tie. The same observation is true if we compare the distances algorithmic-wise or index-wise.

\subsection{Distance matrices}
In this section we will investigate the CM distance matrices for real-world data sets. Our goal is to demonstrate that the CM distance produces interesting and interpretable results.

We calculated the distance matrices using the feature sets \ftname{ind}, \ftname{cov}, and \ftname{freq}. The matrices are given in Figures~\ref{fig:distances}~and~\ref{fig:distances2}. In addition, we computed performance indices, a ratio of the mean of the intra-cluster distances and the mean of the inter-cluster distances, for various clusterings and compare these indices to the ones obtained from the base distances. The results are given in Table~\ref{tab:clusterings}.

\begin{table}[htb!]
\centering
\begin{tabular}{rlrrrrr}
\toprule
 & & & \multicolumn{2}{c}{\ftname{cov}} & \multicolumn{2}{c}{\ftname{freq}} \\
\cmidrule{4-7}
Data & Clustering & ind & $d_{CM}$ & $d_U$ & $d_{CM}$ & $d_U$\\
\midrule
\dtname{Bible} & Old Test. $\mid$ New Test. & $0.79$ & $0.79$ & $0.82$ & $0.73$ & $0.81$ \\
 & Old Test. $\mid$ Gospels $\mid$ Epistles & $0.79$ & $0.79$ & $0.81$ & $0.73$ & $0.81$ \\
\dtname{Addresses} & 1--32 $\mid$ 33--55 & $0.79$ & $0.80$ & $0.85$ & $0.70$ & $0.84$ \\
 & 1--11 $\mid$ 12--22 $\mid$ 23--33 $\mid$ 34--44 $\mid$ 45--55 & $0.83$ & $0.83$ & $0.87$ & $0.75$ & $0.87$ \\
\dtname{Beatles} & 1,2,4--6 $\mid$ 7--10,12--13 $\mid$ 3 $\mid$ 11 & $0.83$ & $0.86$ & $0.83$ & $0.88$ & $0.61$ \\
 & 1,2,4,12,13 $\mid$ 5--10 $\mid$ 3 $\mid$ 11 & $0.84$ & $0.85$ & $0.84$ & $0.89$ & $0.63$ \\
\dtname{20Newsgroups} & Rel.,Pol. $\mid$ Rest & $0.76$ & $0.77$ & $0.67$ & $0.56$ & $0.62$ \\
 & Rel.,Pol. $\mid$ Comp., misc $\mid$ Rest & $0.78$ & $0.78$ & $0.79$ & $0.53$ & $0.79$ \\
\dtname{TopGenres} & Act.,Adv., SciFi $\mid$ Rest & $0.74$ & $0.73$ & $0.64$ & $0.50$ & $0.32$ \\
\dtname{TopDecades} & 1930--1960 $\mid$ 1970--2000 & $0.84$ & $0.83$ & $0.88$ & $0.75$ & $0.88$ \\
 & 1930--1950 $\mid$ 1960--2000 & $0.88$ & $0.88$ & $0.98$ & $0.57$ & $1.06$ \\
\bottomrule
\end{tabular}
\caption{Statistics of various interpretable clusterings. The proportions are the averages of the intra-cluster distances divided by the averages of the inter-cluster distances. Hence small fractions imply tight clusterings.}
\label{tab:clusterings}
\end{table}

\begin{figure}[htb!]
\center
\small
% 20Newsgroups
\begin{minipage}{4.5cm}
\center
\includegraphics[width=4cm]{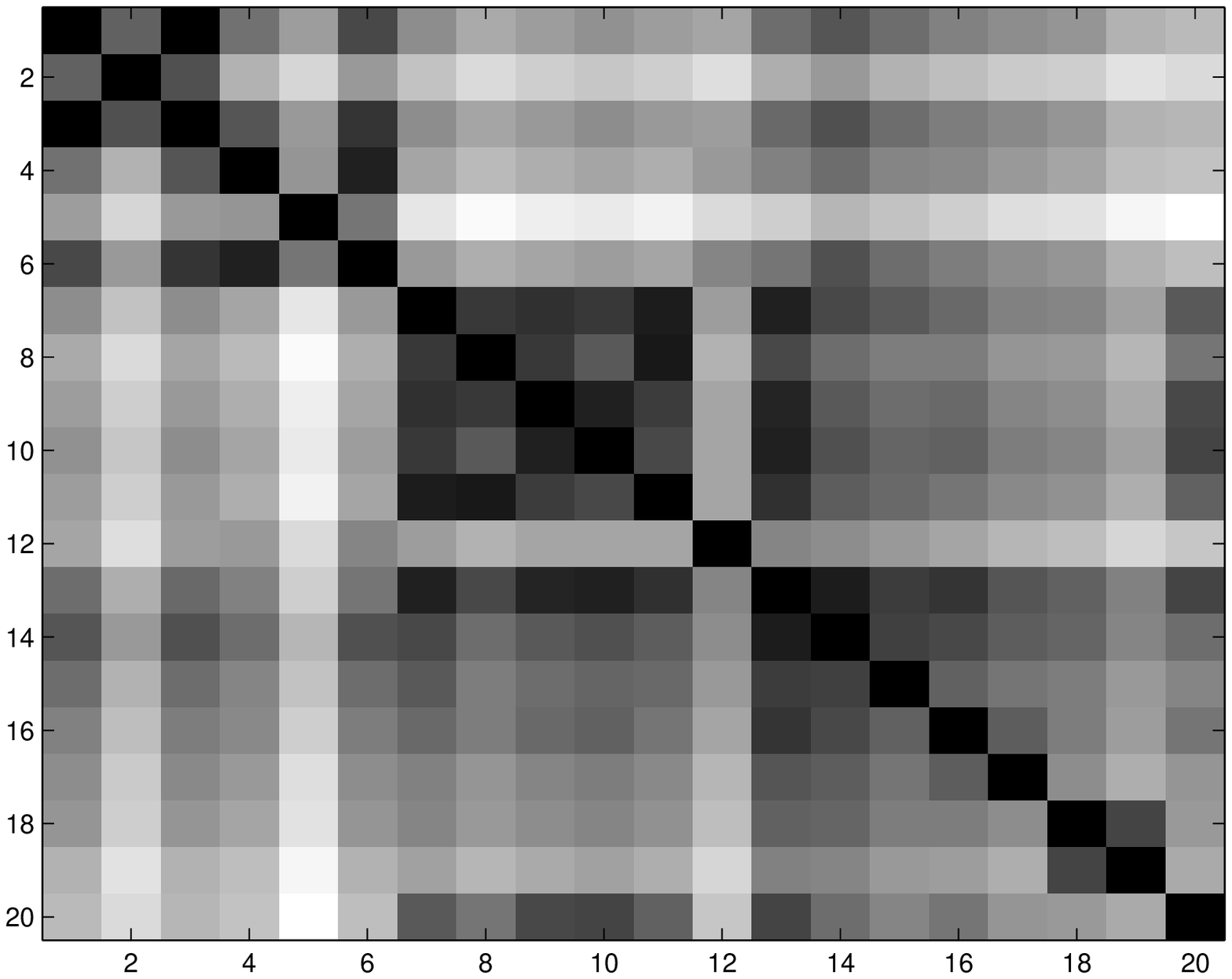}
\dtname{20Newsgroups}, $\dcm{ind}$
\end{minipage}
\begin{minipage}{4.5cm}
\center
\includegraphics[width=4cm]{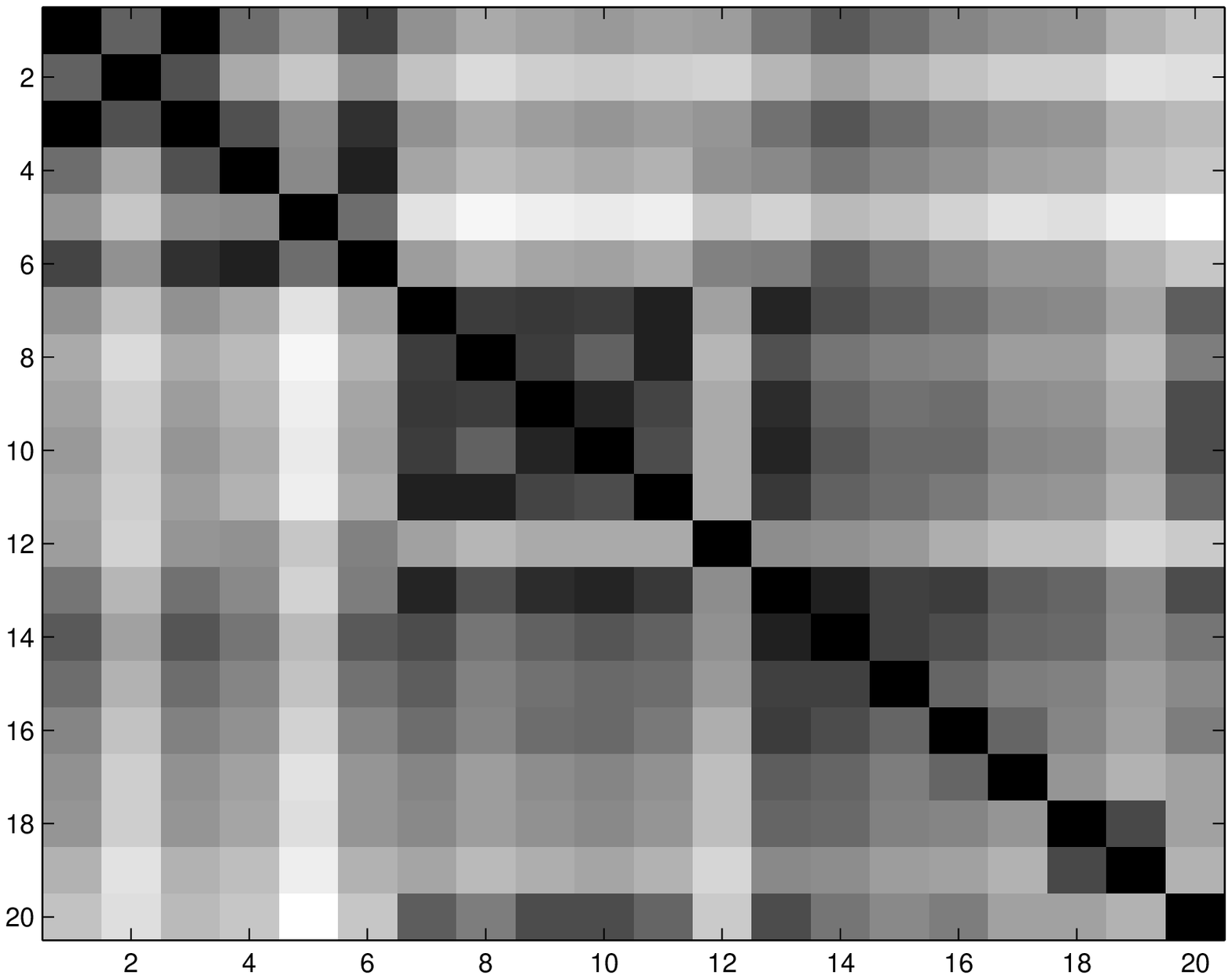}
\dtname{20Newsgroups}, $\dcm{cov}$
\end{minipage}
\begin{minipage}{4.5cm}
\center
\includegraphics[width=4cm]{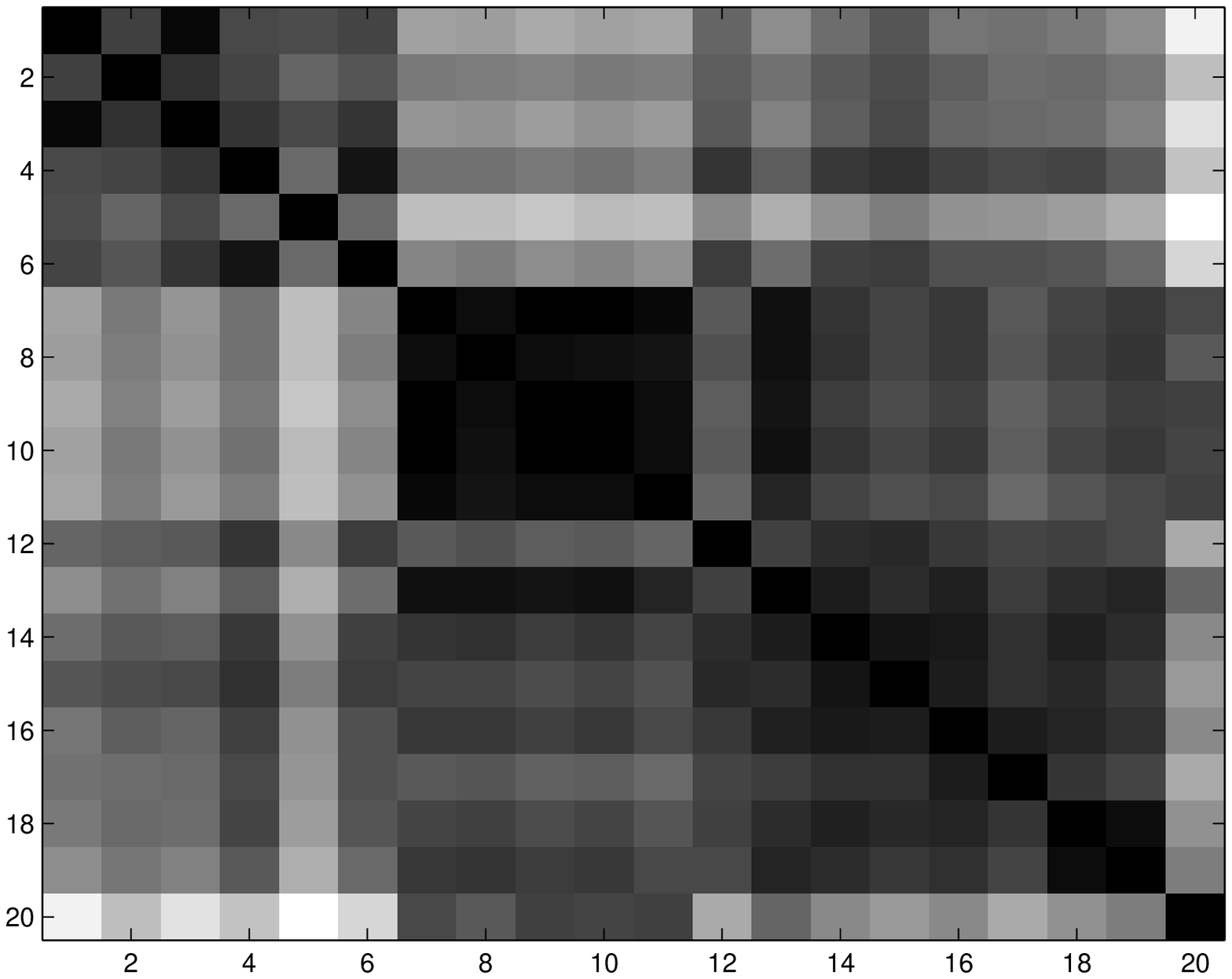}
\dtname{20Newsgroups}, $\dcm{freq}$
\end{minipage}
\bigskip

% TopGenres
\begin{minipage}{4.5cm}
\center
\includegraphics[width=4cm]{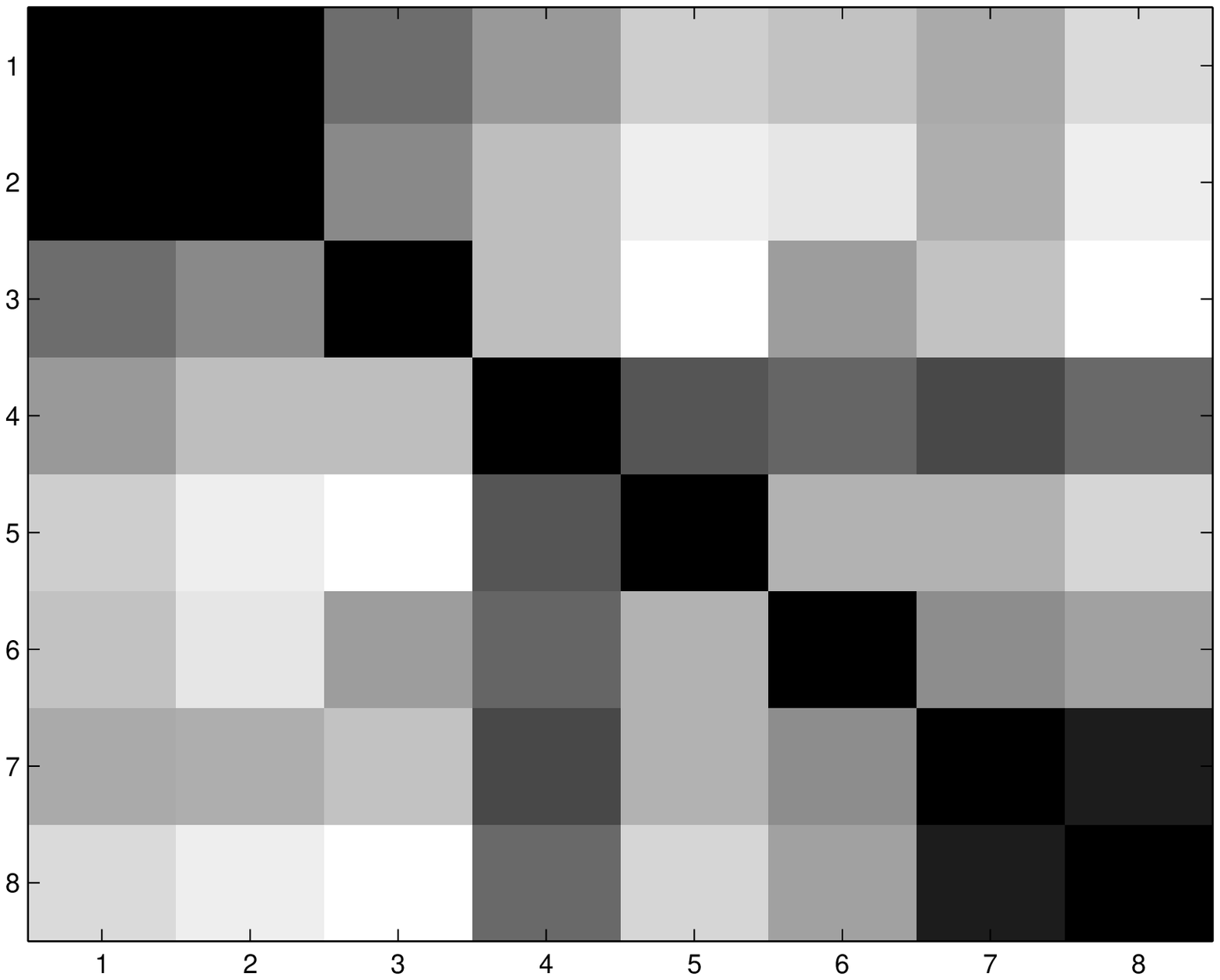}
\dtname{TopGenres}, $\dcm{ind}$
\end{minipage}
\begin{minipage}{4.5cm}
\center
\includegraphics[width=4cm]{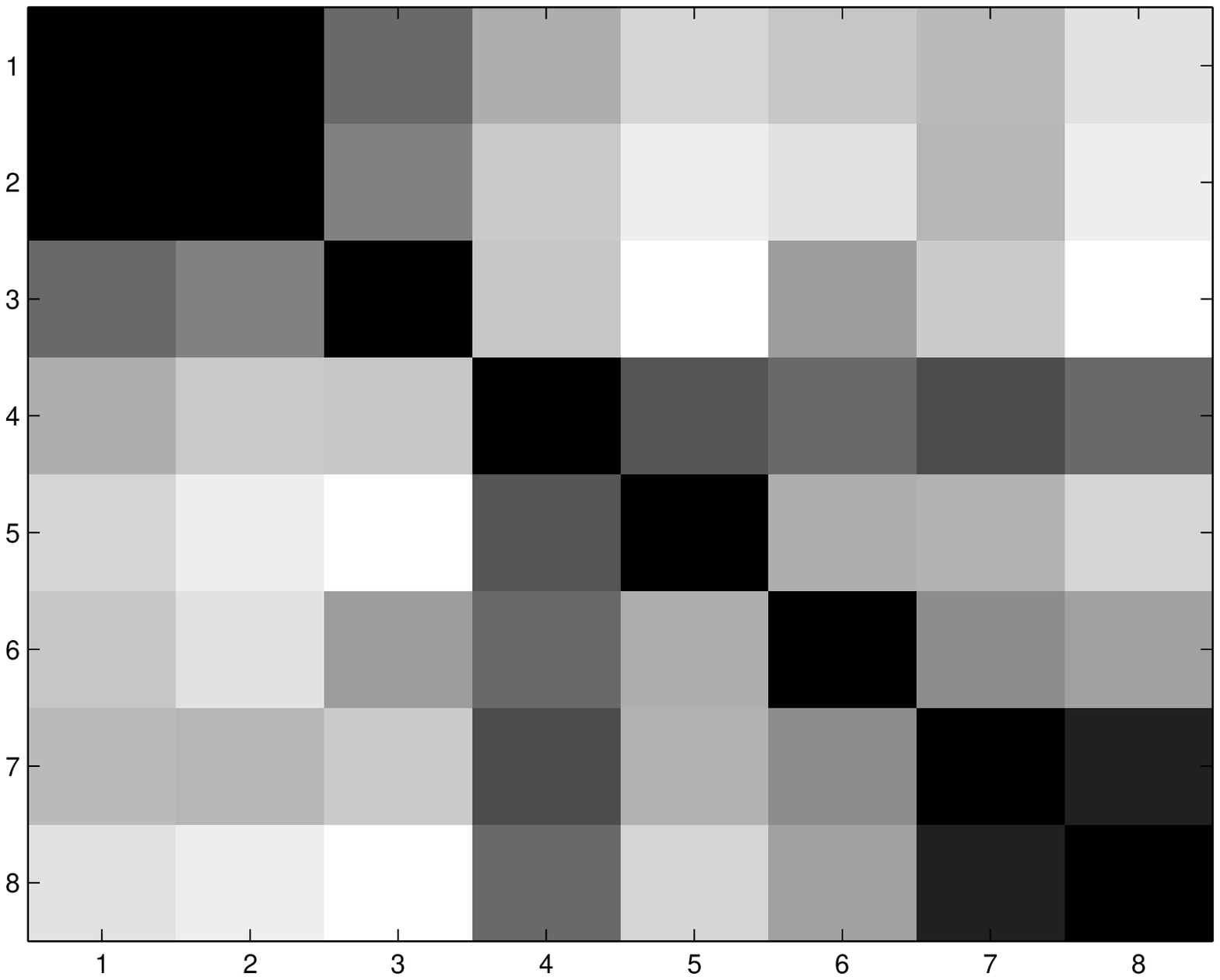}
\dtname{TopGenres}, $\dcm{cov}$
\end{minipage}
\begin{minipage}{4.5cm}
\center
\includegraphics[width=4cm]{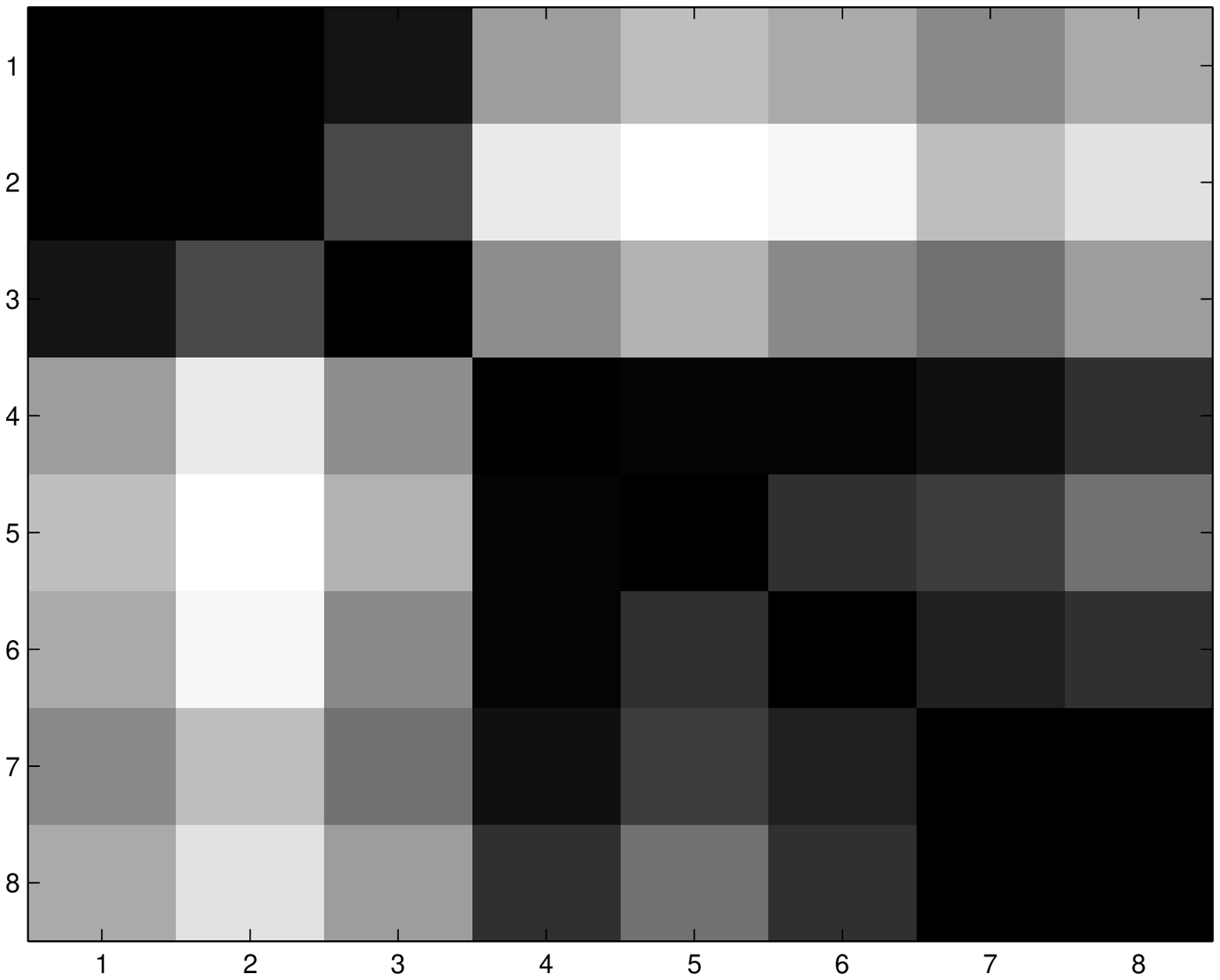}
\dtname{TopGenres}, $\dcm{freq}$
\end{minipage}
\bigskip

% TopDecages
\begin{minipage}{4.5cm}
\center
\includegraphics[width=4cm]{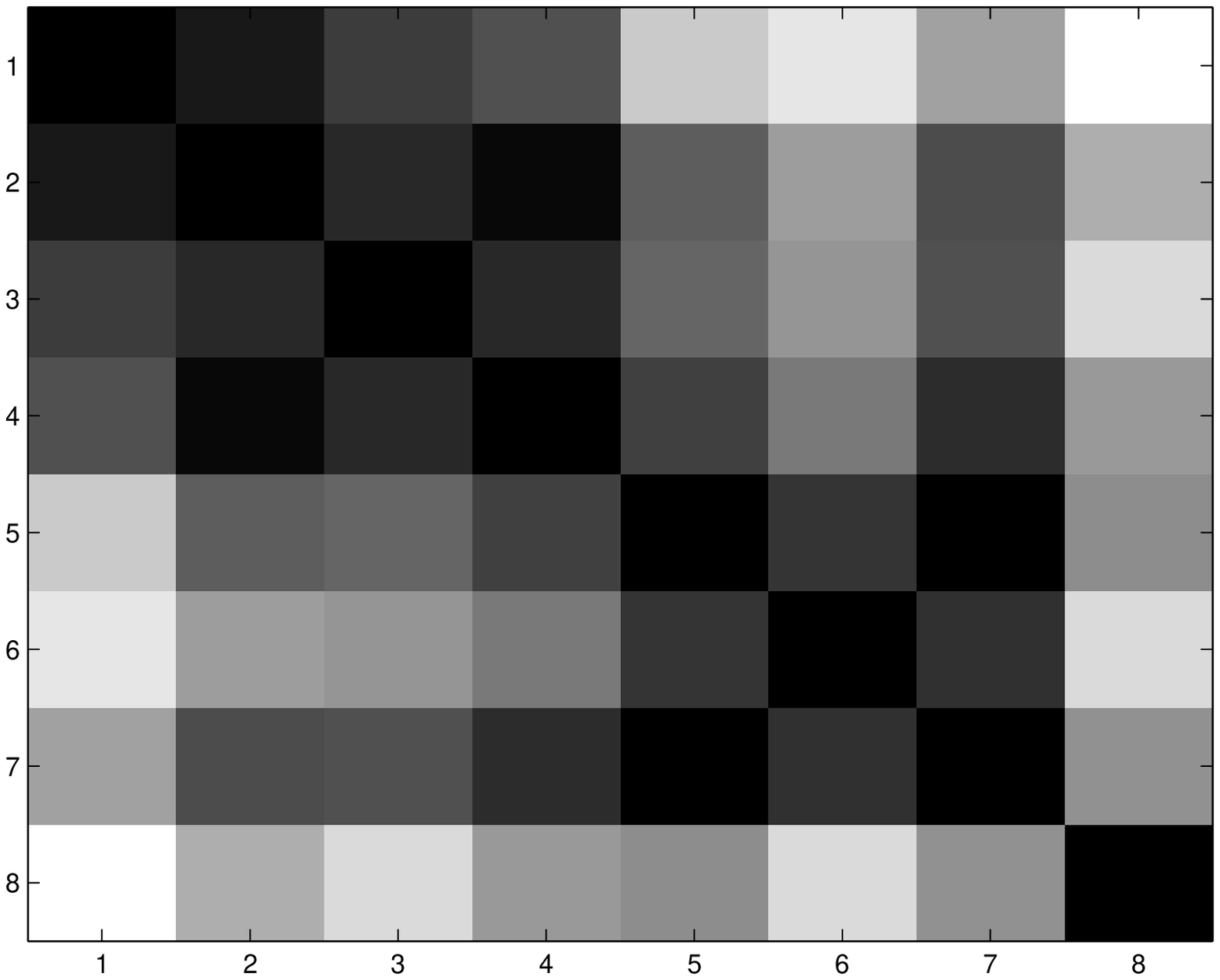}
\dtname{TopDecades}, $\dcm{ind}$
\end{minipage}
\begin{minipage}{4.5cm}
\center
\includegraphics[width=4cm]{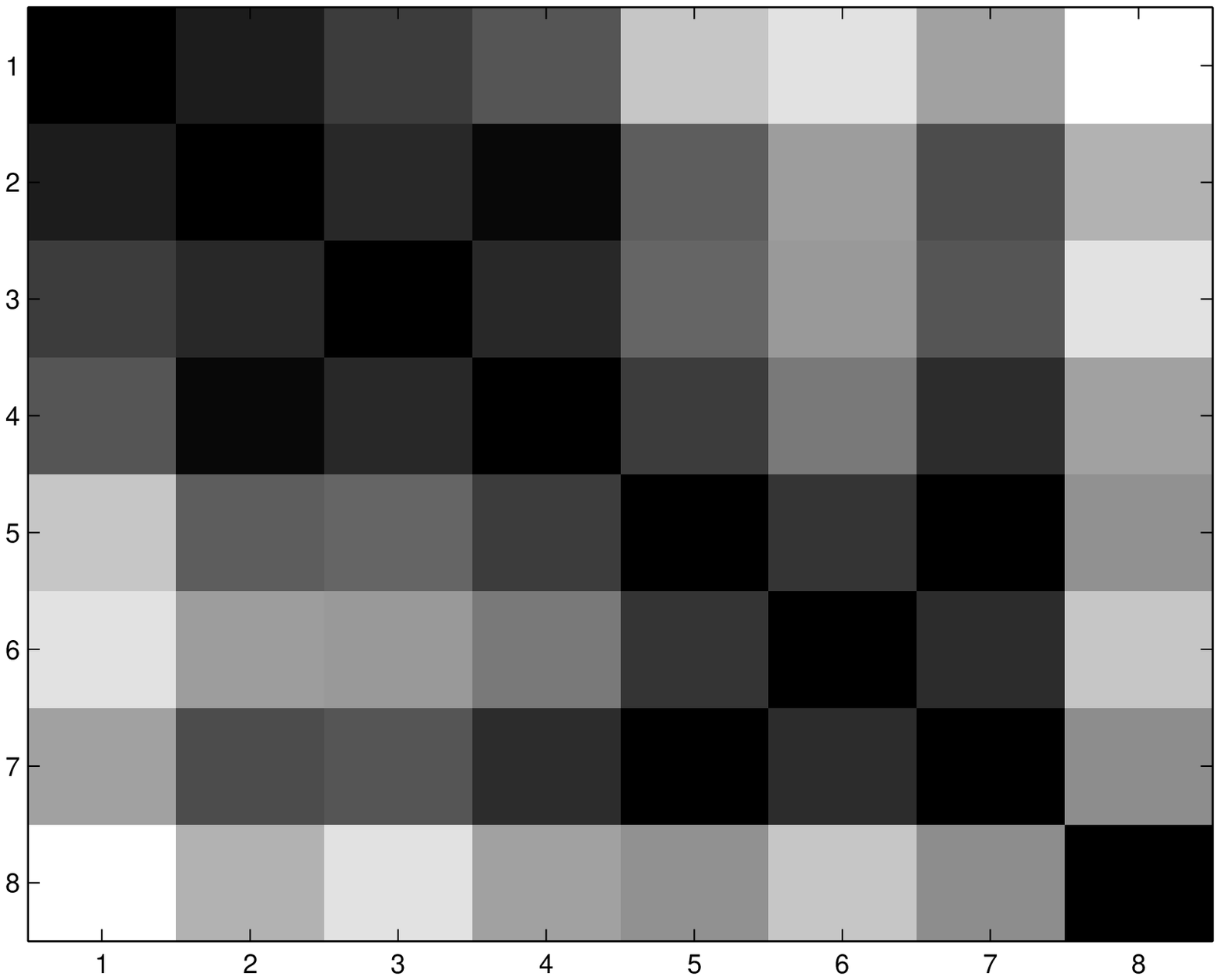}
\dtname{TopDecades}, $\dcm{cov}$
\end{minipage}
\begin{minipage}{4.5cm}
\center
\includegraphics[width=4cm]{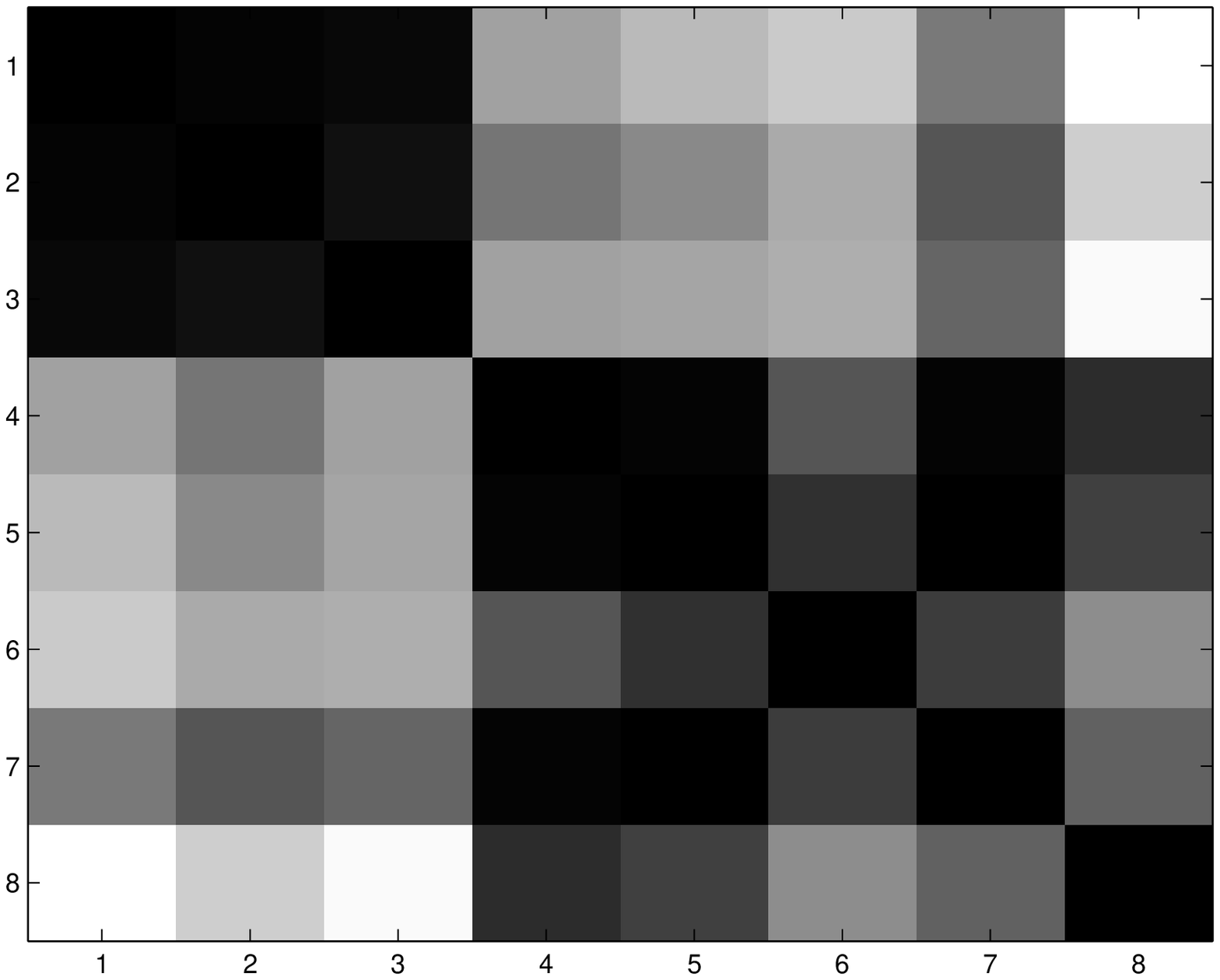}
\dtname{TopDecades}, $\dcm{freq}$
\end{minipage}

% Abstracts
\begin{minipage}{4.5cm}
\center
\includegraphics[width=4cm]{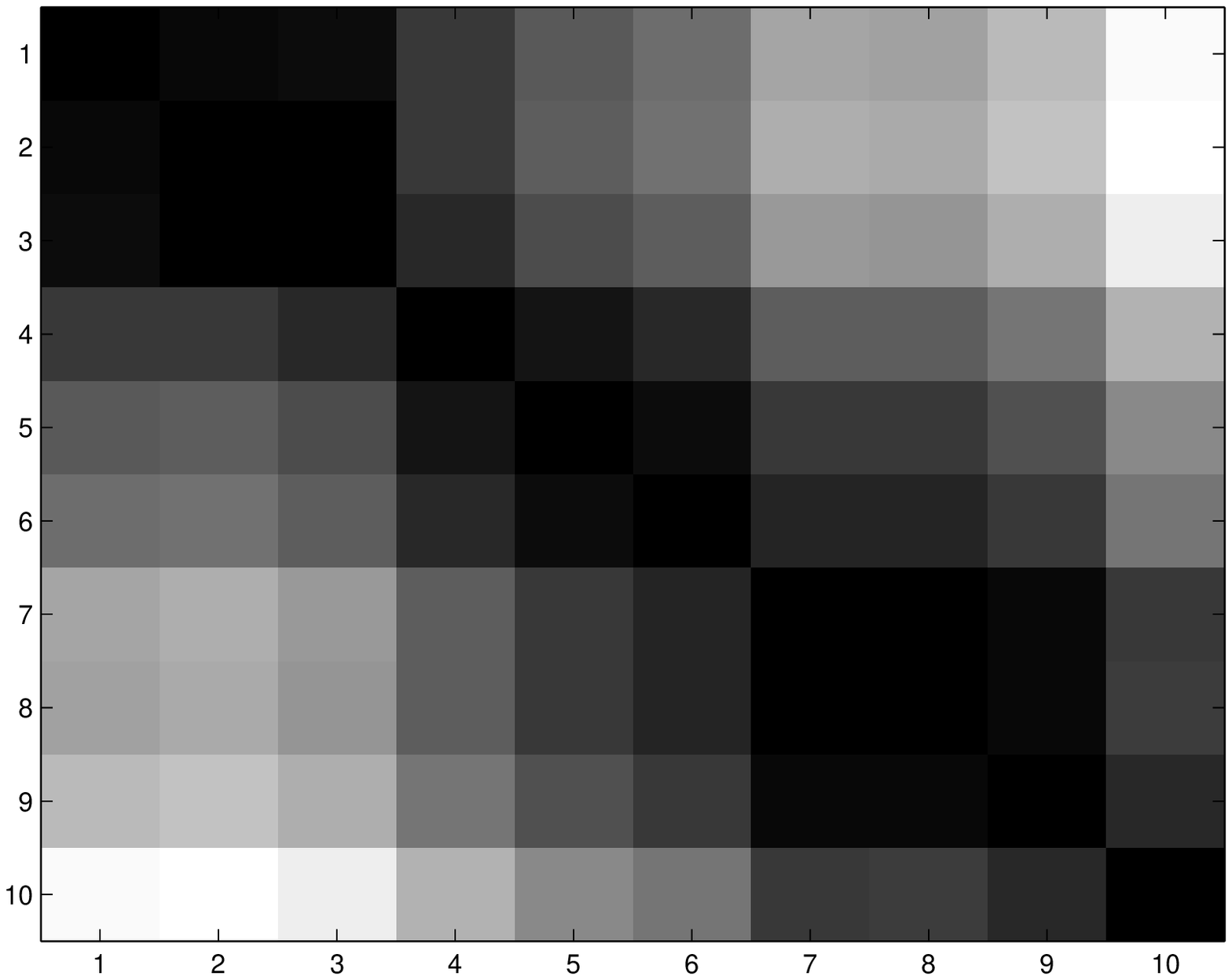}
\dtname{Abstract}, $\dcm{ind}$
\end{minipage}
\begin{minipage}{4.5cm}
\center
\includegraphics[width=4cm]{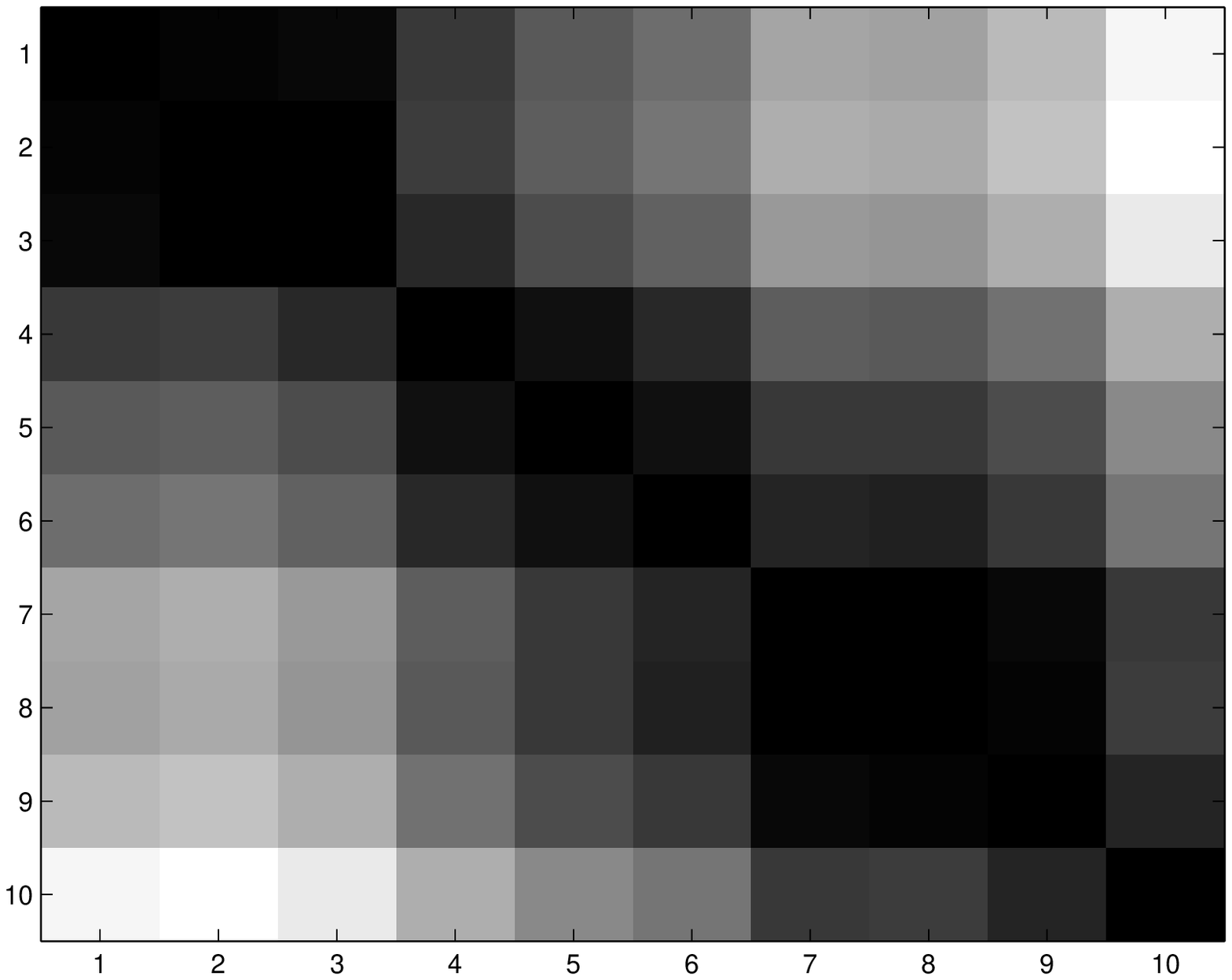}
\dtname{Abstract}, $\dcm{cov}$
\end{minipage}
\begin{minipage}{4.5cm}
\center
\includegraphics[width=4cm]{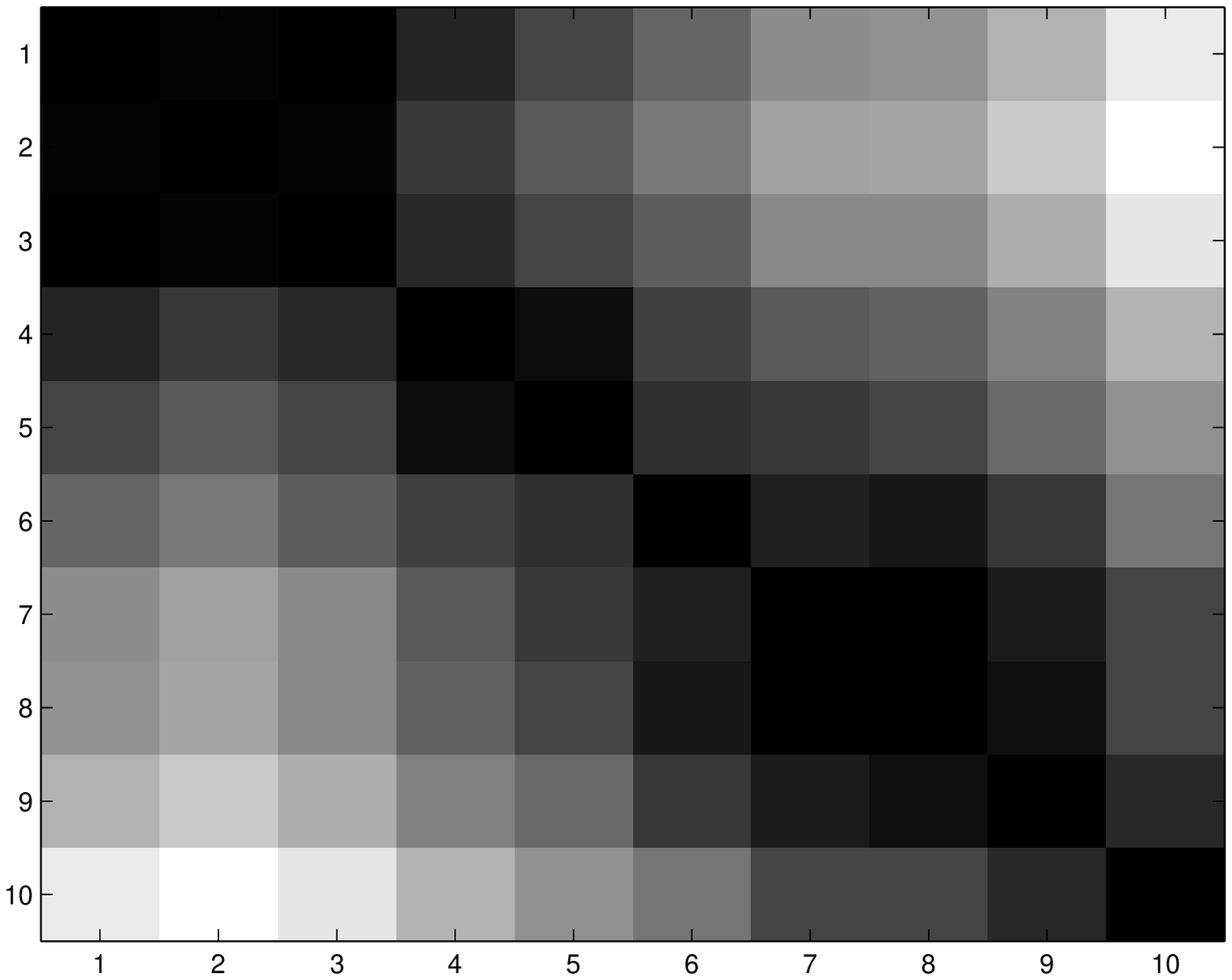}
\dtname{Abstract}, $\dcm{freq}$
\end{minipage}
\caption{Distance matrices for \ftname{20Newsgroups}, \ftname{TopGenres}, \ftname{TopDecades}, and \ftname{Abstract}. In the first column the feature set \ftname{ind} contains the independent means, in the second feature set \ftname{cov} the pairwise correlation is added, and in the third column the feature set \ftname{freq} consists of $10K$ most frequent itemsets, where $K$ is the number of attributes. Darker colours indicate smaller distances.}
\label{fig:distances2}
\end{figure}

\begin{figure}[htb!]
\center
\small
% Bible
\begin{minipage}{5cm}
\center
\includegraphics[width=5cm]{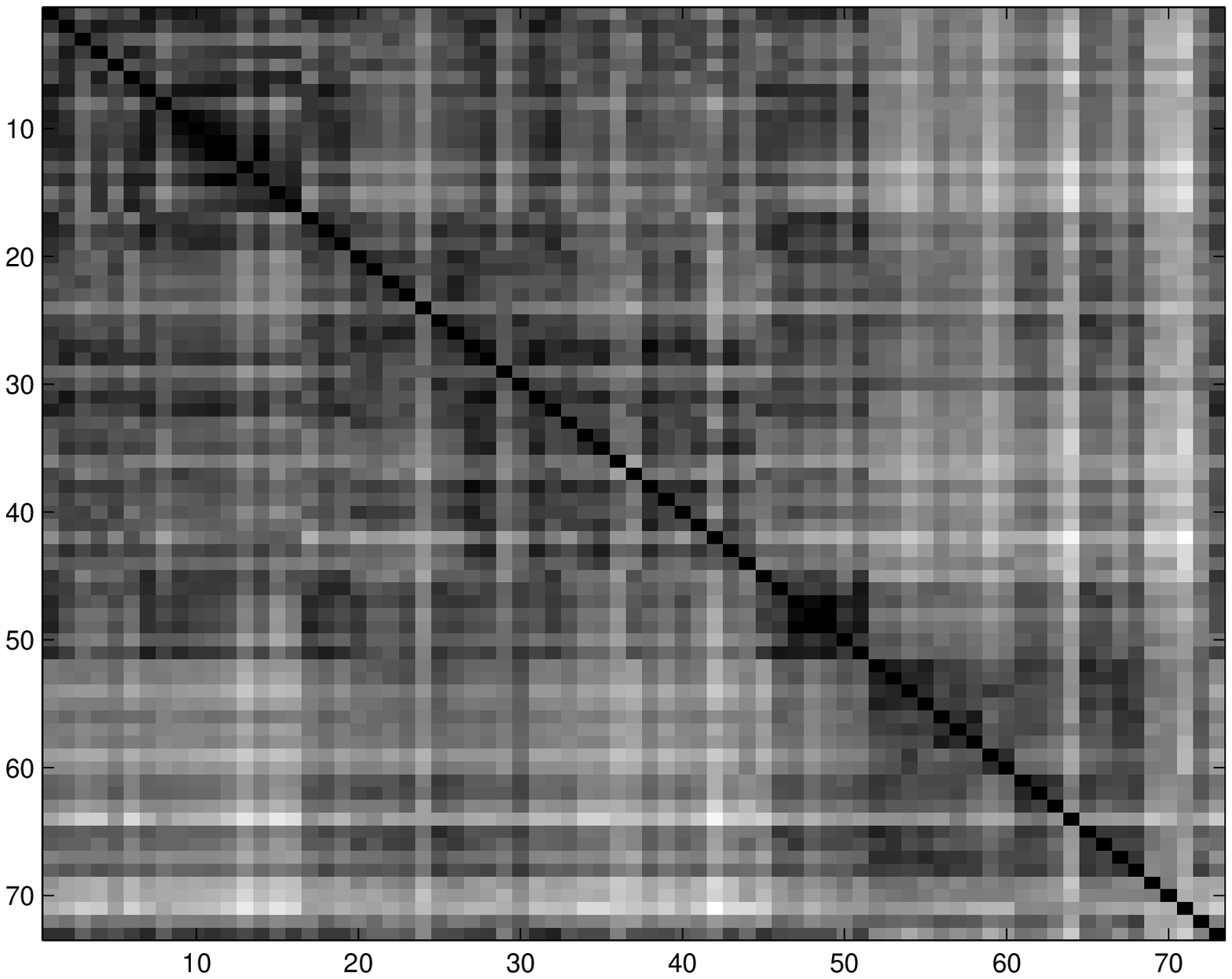}
\dtname{Bible}, $\dcm{ind}$
\end{minipage}
\begin{minipage}{5cm}
\center
\includegraphics[width=5cm]{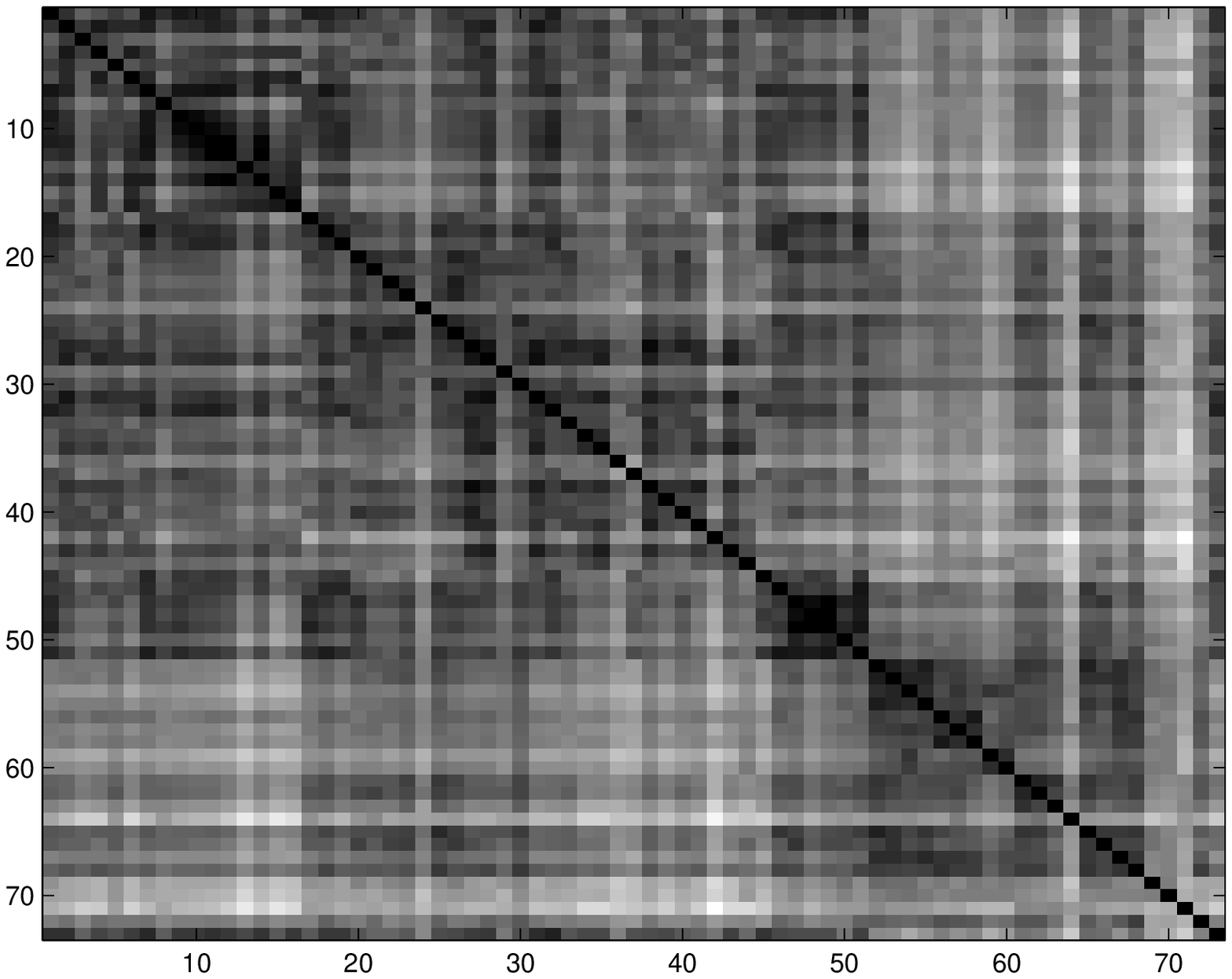}
\dtname{Bible}, $\dcm{cov}$
\end{minipage}
\begin{minipage}{5cm}
\center
\includegraphics[width=5cm]{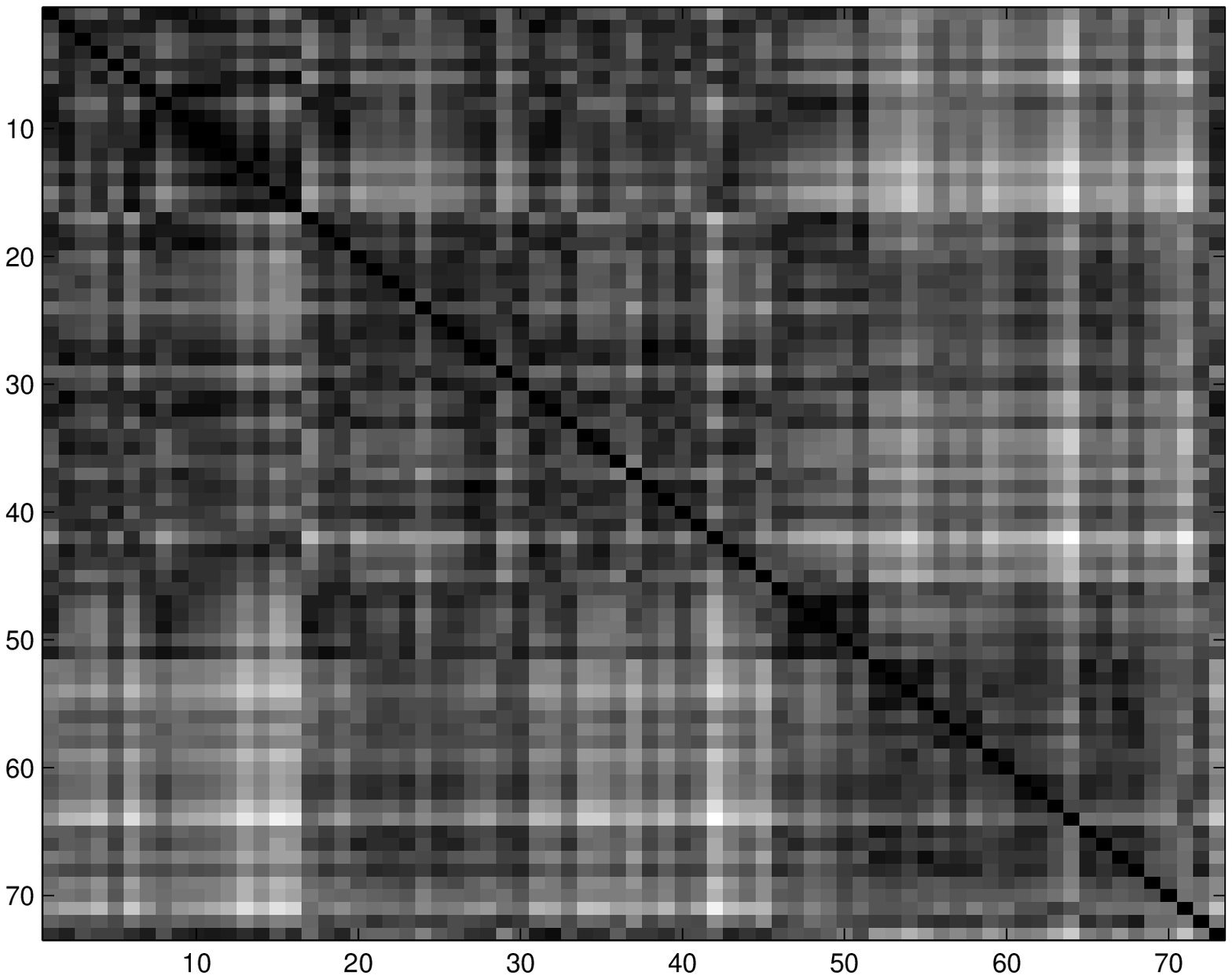}
\dtname{Bible}, $\dcm{freq}$
\end{minipage}
\bigskip

% Addresses
\begin{minipage}{5cm}
\center
\includegraphics[width=5cm]{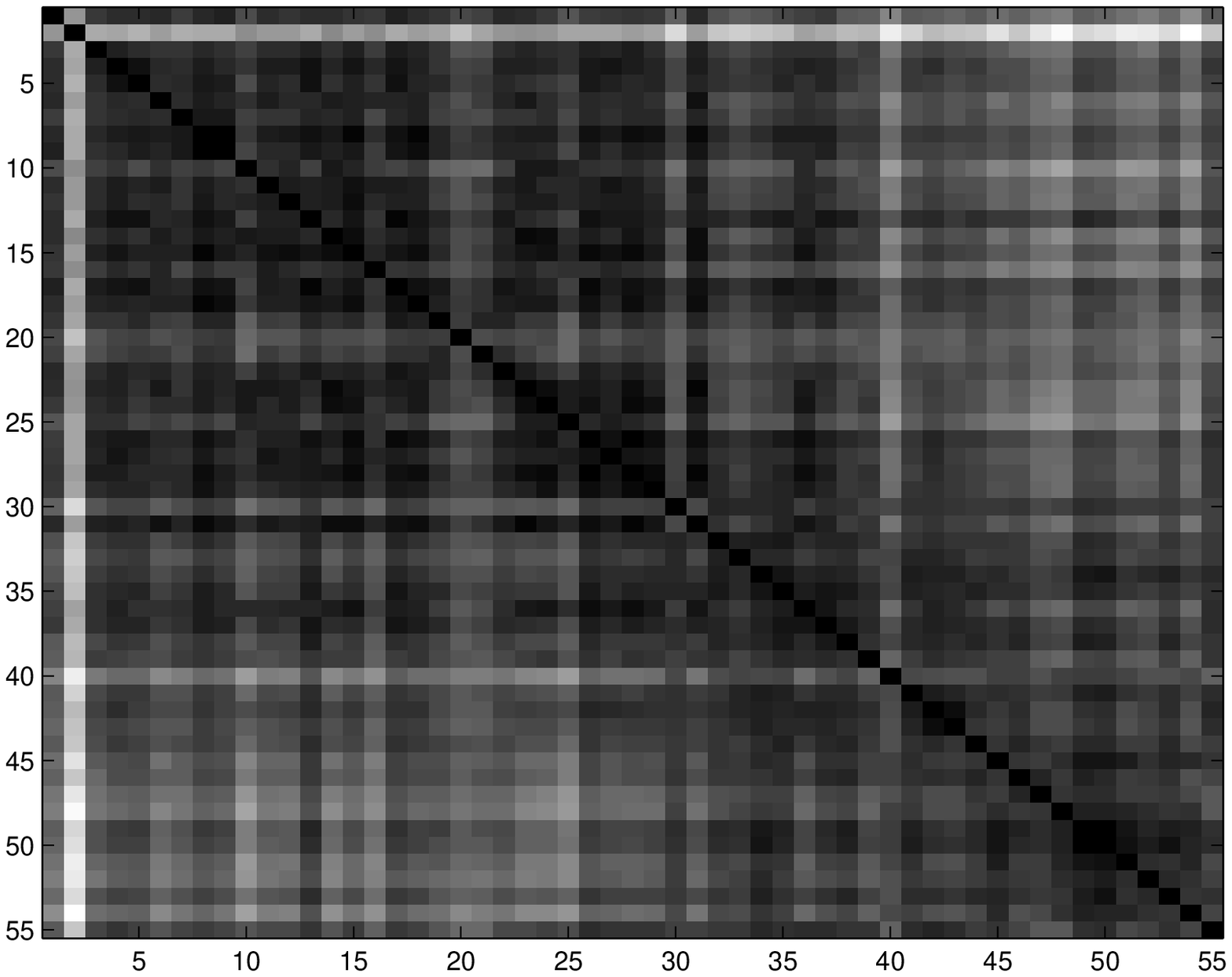}
\dtname{Addresses}, $\dcm{ind}$
\end{minipage}
\begin{minipage}{5cm}
\center
\includegraphics[width=5cm]{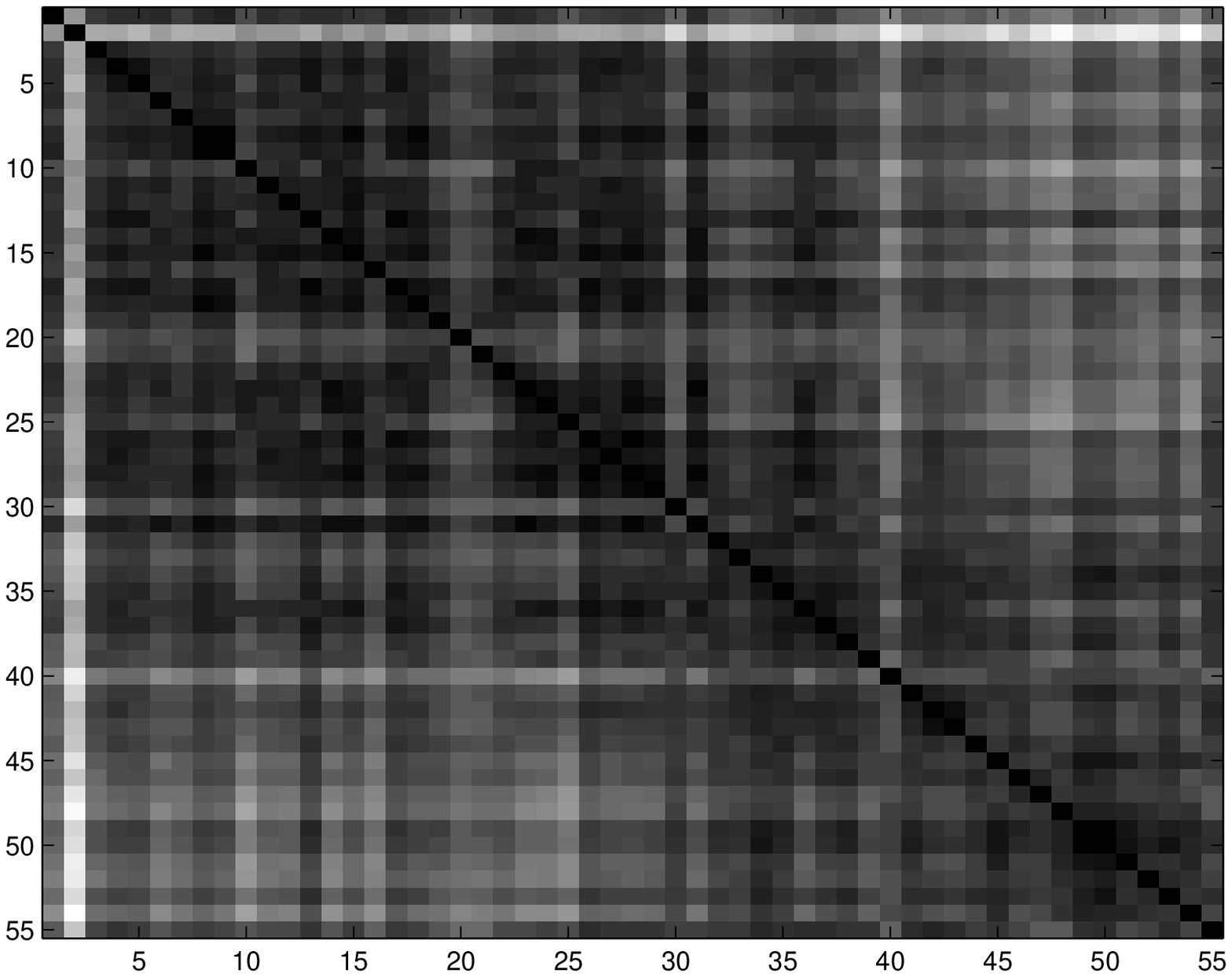}
\dtname{Addresses}, $\dcm{cov}$
\end{minipage}
\begin{minipage}{5cm}
\center
\includegraphics[width=5cm]{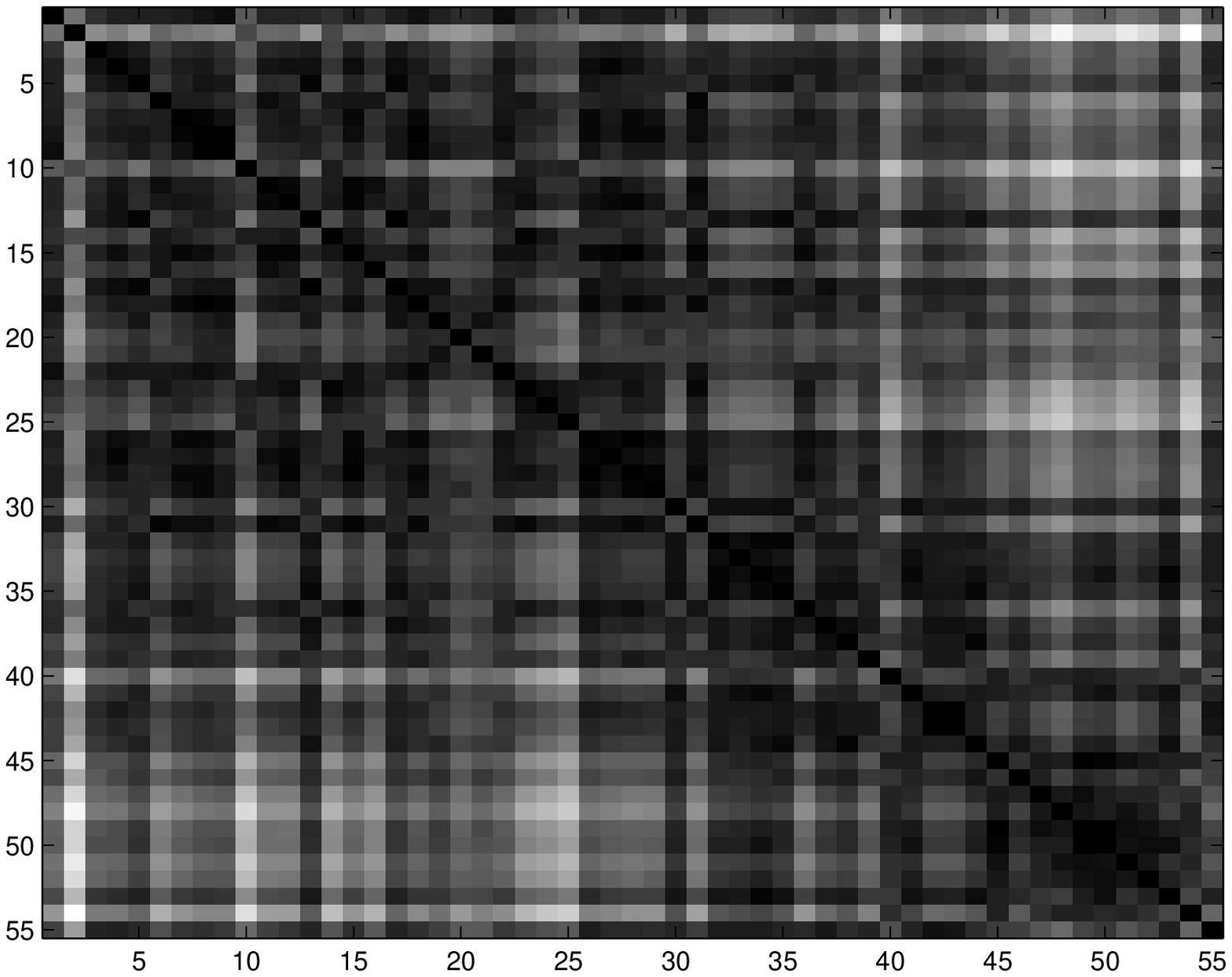}
\dtname{Addresses}, $\dcm{freq}$
\end{minipage}
\bigskip

% Beatles
\begin{minipage}{5cm}
\center
\includegraphics[width=4cm]{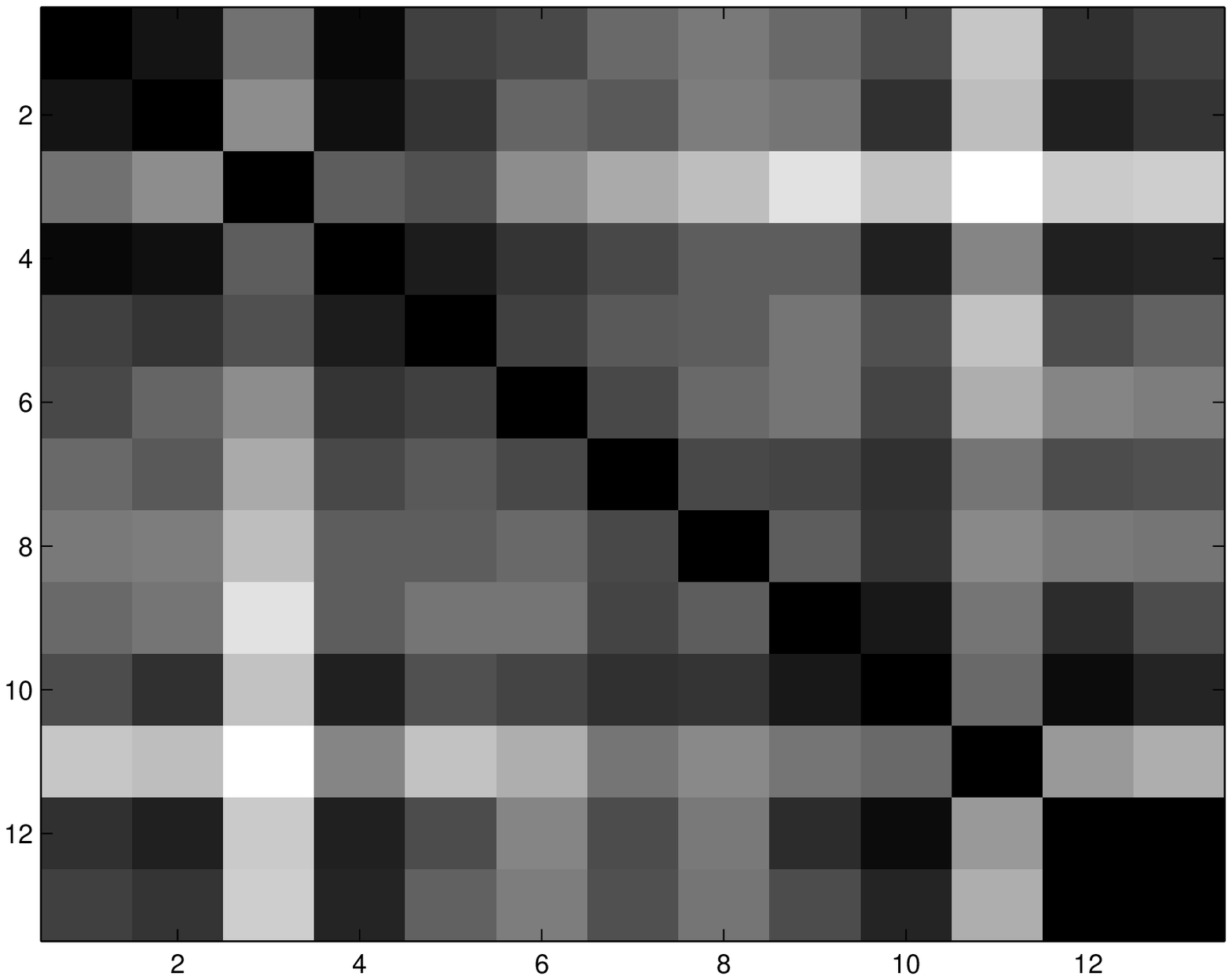}
\dtname{Beatles}, $\dcm{ind}$
\end{minipage}
\begin{minipage}{5cm}
\center
\includegraphics[width=4cm]{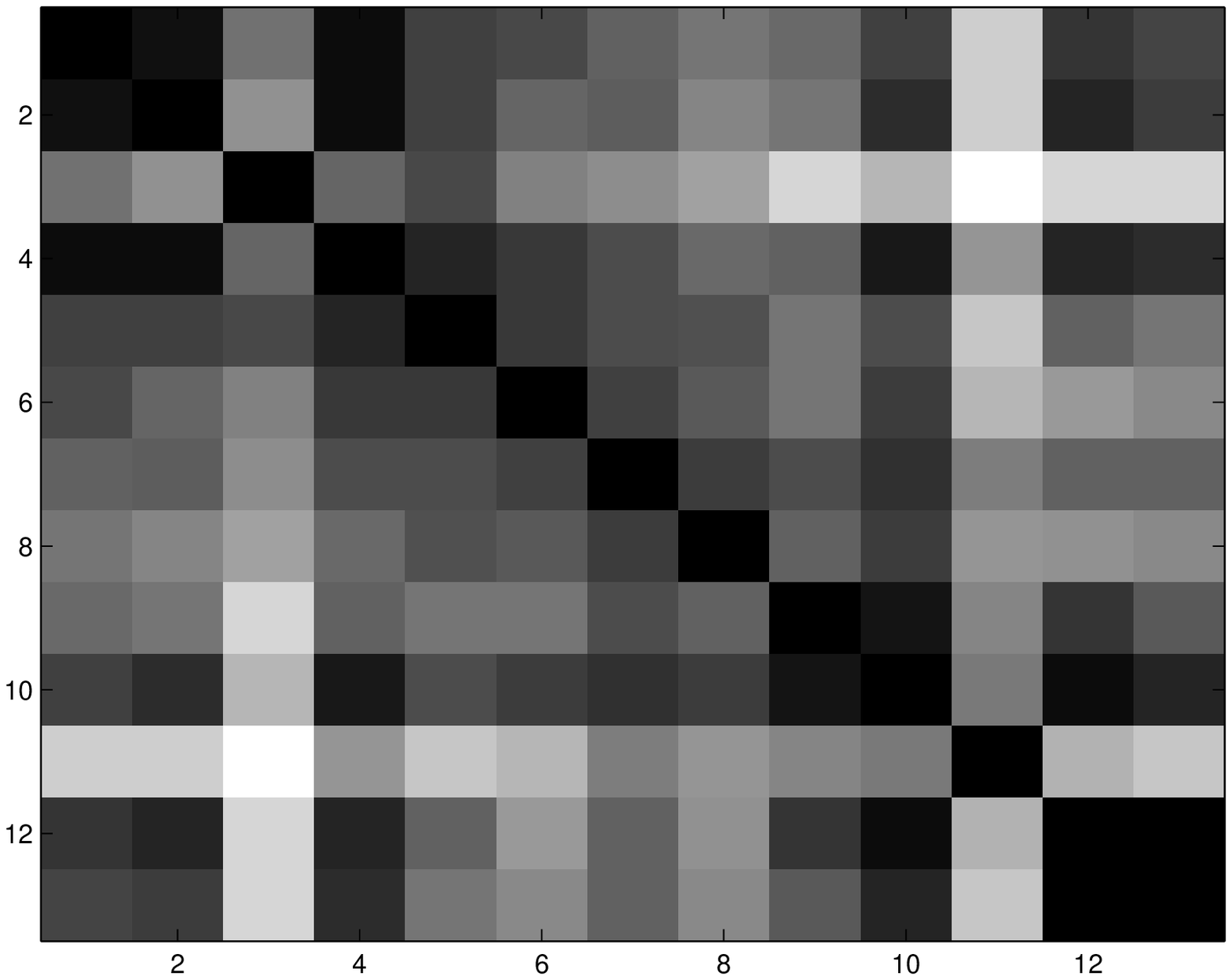}
\dtname{Beatles}, $\dcm{cov}$
\end{minipage}
\begin{minipage}{5cm}
\center
\includegraphics[width=4cm]{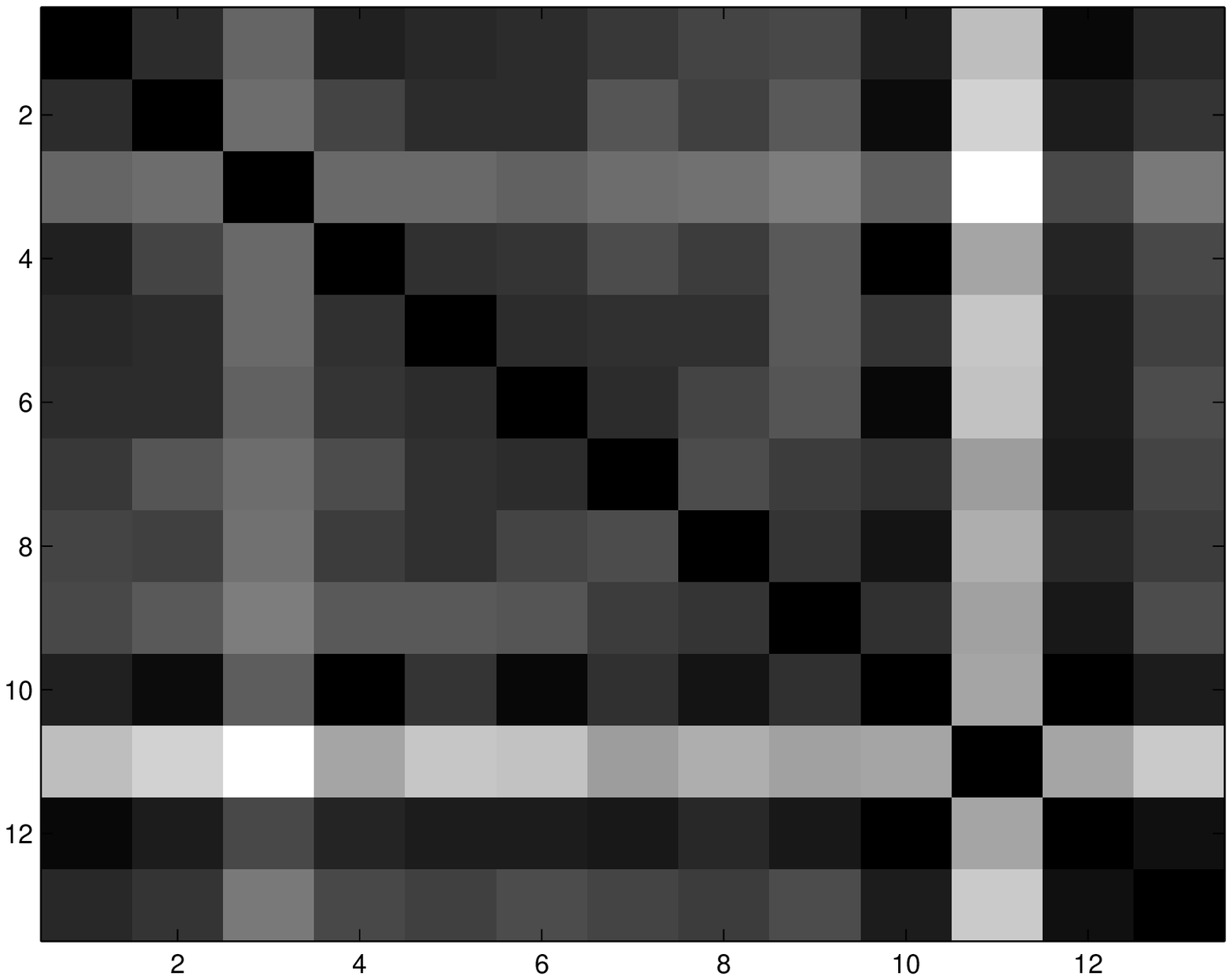}
\dtname{Beatles}, $\dcm{freq}$
\end{minipage}
\caption{Distance matrices for \dtname{Bible}, \dtname{Addresses}, and \dtname{Beatles}. In the first column the feature set \ftname{ind} contains the independent means, in the second feature set \ftname{cov} the pairwise correlation is added, and in the third column the feature set \ftname{freq} consists of $10K$ most frequent itemsets, where $K$ is the number of attributes. Darker colours indicate smaller distances.}
\label{fig:distances}
\end{figure}

%We see that distances based on $S_\iset{I}$ are almost identical (up to normalisation) to the distances based on $S_\iset{C}$. However, the distances based on frequent itemsets $S_\sigma$ are somewhat different. One explanation may be that in $S_\sigma$ some attributes tend to be emphasised whereas in $S_\iset{I}$ and in $S_\iset{C}$ all attributes are treated equally.

We should stress that standard edit distances would not work in these data setups. For example, the sequences have different lengths and hence Levenshtein distance cannot work.

The imperative observation is that, according to the CM distance, the data sets have structure. We can also provide some interpretations to the results: In \dtname{Bible} we see a cluster starting from the $46$th book. The New Testament starts from the $47$th book. An alternative clustering is obtained by separating the Epistles, starting from the $52$th book, from the Gospels. In \dtname{Addresses} we some temporal dependence. Early speeches are different than the modern speeches. In \dtname{Beatles} we see that the early albums are linked together and the last two albums are also linked together. The third album, \dtname{Help!}, is peculiar. It is not linked to the early albums but rather to the later work. One explanation may be that, unlike the other early albums, this album does not contain any cover songs. In \dtname{20Newsgroups} the groups of politics and of religions are close to each other and so are the computer-related groups. The group \dtname{misc.forsale} is close to the computer-related groups. In \dtname{TopGenres} \dtname{Action} and \dtname{Adventure} are close to each other. Also \dtname{Comedy} and \dtname{Romance} are linked. In \dtname{TopDecades} and in \dtname{Abstract} we see temporal behaviour. In Table~\ref{tab:clusterings} the CM distance outperforms the base distance, except for \dtname{Beatles} and \dtname{TopGenres}.

%% file: conclusions.tex
\section{Conclusions and Discussion}
\label{sec:conclusions}
Our task was to find a versatile distance that has nice statistical properties and that can be evaluated efficiently. The CM distance fulfils our goals. In theoretical sections we proved that this distance takes properly into account the correlation between features, and that it is the only (Mahalanobis) distance that does so. Even though our theoretical justifications are complex, the CM distance itself is rather simple. In its simplest form, it is the $L_2$ distance between the means of the individual attributes. On the other hand, the CM distance has a surprising form when the features are itemsets.

In general, the computation time of the CM distance depends of the size of sample space that can be exponentially large. Still, there are many types of feature functions for which the distance can be solved. For instance, if the features are itemsets, then the distance can be solved in polynomial time. In addition, if the itemsets form an antimonotonic family, then the distance can be solved in linear time.

In empirical tests the CM distance implied that the used data sets have structure, as expected. The performance of the CM distance compared to the base distance depended heavily on the data set. We also showed that the feature sets \ftname{ind} and \ftname{cov} produced almost equivalent distances, whereas using frequent itemsets produced very different distances.

Sophisticated feature selection methods were not compared in this paper. Instead, we either decided explicitly the set of features or deduced them using \alname{Apriori}. We argued that we cannot use the traditional approaches for selecting features of data sets, unless we are provided some additional information.

%% file: appendix.tex
\section{}
In this section we will prove the theorems given in this paper.
\subsection{Proof of Theorem~\ref{thr:calc}}
\label{prf:calc}
To simplify the notation denote $S_0(x) = 1$, $\theta^*_1 = \envec{1, \theta_{11}}{\theta_{1N}}$ and $\theta^*_2 = \envec{1, \theta_{21}}{\theta_{2N}}$.
The norm function restricted to the affine space has one minimum and it can be found using Lagrange multipliers. Thus we can express the vectors $u_i$ in Eq.~\ref{eq:def}
\[
u_{ij} = \lambda_i^TS(j),
\]
where $j \in \Omega$ and $\lambda_i$ is the column vector of length $N+1$ consisting of the corresponding Lagrange multipliers.
The distance is equal to
\[
\begin{split}
\dist{D_1}{D_2}{S}^2 & = \abs{\Omega}\norm{u_1-u_2}^2 \\
& = \abs{\Omega}\sum_{j \in \Omega}\pr{u_{1j}-u_{2j}}\pr{u_{1j}-u_{2j}} \\
& = \abs{\Omega}\sum_{j \in \Omega}\pr{u_{1j}-u_{2j}}\pr{\lambda_1^TS(j)-\lambda_2^TS(j)} \\
& = \abs{\Omega}\pr{\lambda_1-\lambda_2}^T\sum_{j \in \Omega}\pr{u_{1j}-u_{2j}}S(j) \\
 & = \abs{\Omega} \pr{\lambda_1-\lambda_2}^T\pr{\theta_1^*-\theta_2^*}.
\end{split}
\]
Since $u_i \in \const{S}{\theta_i}$, the multipliers $\lambda_i$ can be solved from the equation
\[
\theta_i^* = \sum_{j \in \Omega} S(j) u_{ij} = \sum_{j \in \Omega} S(j) \lambda_i^TS(j) = \pr{\sum_{j \in \Omega} S(j)S(j)^T} \lambda_i,
\]
i.e., $\theta^*_i = A\lambda_i$, where $A$ is an $(N+1)\times (N+1)$ matrix $A_{xy} = \sum_j S_x(j)S_y(j)$. It is straightforward to prove that the existence of $\cov{}{-1}{S}$ implies that $A$ is also invertible. Let $B$ be an $N\times N$ matrix formed from $A^{-1}$ by removing the first row and the first column. We have
\[
\begin{split}
\abs{\Omega}\norm{u_1-u_2}^2 & = \abs{\Omega}\pr{\theta_1^*-\theta_2^*}^TA^{-1}\pr{\theta_1^*-\theta_2^*} \\
& = \abs{\Omega}\pr{\theta_1-\theta_2}^TB\pr{\theta_1-\theta_2}.
\end{split}
\]
The last equality is true since $\theta^*_{10} = \theta^*_{20}$.

We need to prove that $\abs{\Omega}B = \cov{}{-1}{S}$. Let $\spr{c; B}$ be the matrix obtained from $A^{-1}$ by removing the first row. Let $\gamma = \mean{}{S}$ taken with respect to the uniform distribution. Since the first column of $A$ is equal to $\abs{\Omega}\spr{1, \gamma}$, it follows that $c = -B\gamma$. From the identity
\[
c_xA_{(0,y)} + \sum_{z=1}^NB_{(x,z)}A_{(z,y)} = \delta_{xy}
\]
we have
\[
\sum_{z=1}^NB_{(x,z)}\pr{A_{(z,y)}-A_{(0,y)}\gamma_z} = \sum_{z=1}^N\abs{\Omega}B_{(x,z)}\pr{\abs{\Omega}^{-1}A_{(z,y)}-\gamma_y\gamma_z} = \delta_{xy}.
\]
Since $\abs{\Omega}^{-1}A_{(z,y)}-\gamma_z\gamma_y$ is equal to the $(z,y)$ entry of $\cov{}{}{S}$, the theorem follows.
\subsection{Proofs of Theorems given in Section~\ref{sec:properties}}
\begin{proof}[Theorem~\ref{thr:metric}]
The covariance matrix $\cov{}{}{S}$ in Theorem~\ref{thr:calc} depends only on $S$ and is positive definite. Therefore, the CM distance is a Mahalanobis distance.
\end{proof}

\begin{proof}[Theorem~\ref{thr:augdata}]
Let $\theta_i = S(D_i)$ for $i=1,2,3$. The frequencies for $D_1 \cup D_3$ and $D_2 \cup D_3$ are $(1-\epsilon)\theta_1+\epsilon\theta_3$ and $(1-\epsilon)\theta_2+\epsilon\theta_3$, respectively. The theorem follows from Theorem~\ref{thr:calc}.
\end{proof}
The following lemma proves Theorem~\ref{thr:linear}.
\begin{lemma} Let $\funcdef{A}{\real^N}{\real^M}$ and define a function $T(\omega) = A(S(\omega))$. Let $\phi = \freq{T}{D}$ and $\theta = \freq{S}{D}$ be the frequencies for some data set $D$. Assume further that there is no two data sets $D_1$ and $D_2$ such that $\freq{S}{D_1} = \freq{S}{D_2}$ and $\freq{T}{D_1} \neq \freq{T}{D_2}$. Then $\dist{D_1}{D_2}{T} \leq \dist{D_1}{D_2}{S}$. The equality holds if for a fixed $\phi$ the frequency $\theta$ is unique.
\label{thr:transform}
\end{lemma}
Before proving this lemma, let us explain why the uniqueness requirement is needed: Assume that the sample space $\Omega$ consists of two-dimensional binary vectors, that is, 
\[
\Omega = \set{\pr{0,0},\pr{1,0},\pr{0,1},\pr{1,1}}.
\]
We set the features to be $S(\omega) = \vect{\omega_1, \omega_2}$. Define a function $T(x) = \vect{\omega_1, \omega_2, \omega_1\omega_2} = \vect{S_1(\omega), S_2(\omega), S_1(\omega)S_2(\omega)}$. Note that uniqueness assumption is now violated. Without this assumption the lemma would imply that $\dist{D_1}{D_2}{T} \leq \dist{D_1}{D_2}{S}$ which is in general false.
\begin{proof}
Let $\theta_1 = \freq{S}{D_1}$ and $\phi_1 = \freq{T}{D_1}$. Pick $u \in \const{S}{\theta_1}$. The frequency of $S$ taken with the respect to $u$ is $\theta_1$ and because of the assumption the corresponding frequency of $T$ is $\phi_1$. It follows that $\const{S}{\theta_i} \subseteq \const{T}{\phi_i}$. The theorem follows from the fact that the CM distance is the shortest distance between the affine spaces $\const{S}{\theta_1}$ and $\const{S}{\theta_2}$.
\end{proof}

\subsection{Proof of Theorem~\ref{thr:char}}
\label{prf:char}
It suffices to prove that the matrix $C(S)$ is proportional to the covariance matrix $\cov{}{}{S}$. The notation $\delta\pr{\omega_1 \mid \omega_2}$ used in the proof represents a feature function $\funcdef{\delta}{\Omega}{\set{0,1}}$ which returns $1$ if $\omega_1 = \omega_2$ and $0$ otherwise.

Before proving the theorem we should point one technical detail. In general, $C(S)$ may be singular, especially in Assumption~\ref{as:1}. In our proof we will show that $C(S) \propto \cov{}{}{S}$ and this does not require $C(S)$ to be invertible. However, if one wants to evaluate the distance $d$, then one must assume that $C(S)$ is invertible.

Fix indices $i$ and $j$ such that $i \neq j$. Let $T(\omega) = \vect{S_i(\omega), S_j(\omega)}$. If follows from Assumption~\ref{as:1} that
\[
C(T) = \spr{
\begin{array}{cc}
C_{ii}(S) & C_{ij}(S) \\
C_{ji}(S) & C_{jj}(S) \\
\end{array}
}.
\]
This implies that $C_{ij}(S)$ depends only on $S_i$ and $S_j$. In other words, we can say $C_{ij}(S) = C_{ij}(S_i, S_j)$. Let $\funcdef{\rho}{\enset{1}{N}}{\enset{1}{N}}$ be some permutation function and define $U(x) = \envec{S_{\rho(1)}(x)}{S_{\rho(N)}(x)}$. Assumption~\ref{as:1} implies that
\[
C_{\rho(i)\rho(j)}(S) = C_{ij}(U) = C_{ij}(U_i, U_j) = C_{ij}(S_{\rho(i)}, S_{\rho(j)}).
\]
This is possible only if all non-diagonal entries of $C$ have the same form or, in other words, $C_{ij}(S) = C_{ij}(S_i, S_j) = C(S_i, S_j)$. Similarly, the diagonal entry $S_{ii}$ depends only on $S_i$ and all the diagonal entries have the same form $C_{ii}(S) = C(S_i)$. To see the connection between $C(S_i)$ and $C(S_i, S_j)$ let $V(\omega) = \vect{S_i(\omega), S_i(\omega)}$ and let $W(\omega) = \vect{2S_i(\omega)}$. We can represent $W(\omega) = V_1(\omega) + V_2(\omega)$. Now Assumption~\ref{as:1} implies
\[
\begin{split}
4C(S_i) & = C(W) = C(V_{11})+2C(V_{12},V_{21})+C(V_{22}) \\
& = 2C(S_i) + 2C(S_i,S_i)
\end{split}
\]
which shows that $C(S_i) = C(S_i, S_i)$. Fix $S_j$ and note that Assumption~\ref{as:1} implies that $C(S_i, S_j)$ is a linear function of $S_i$. Thus $C$ has a form 
\[
C(S_i, S_j) = \sum_{\omega \in \Omega}S_i(\omega)h(S_j, \omega)
\]
for some specific map $h$. Let $\alpha \in \Omega$. Then $C(\delta\pr{\omega \mid \alpha}, S_j) = h(S_j, \alpha)$ is a linear function of $S_j$. Thus $C$ has a form
\[C(S_i, S_j) = \sum_{\omega_1,\omega_2 \in \Omega}S_i(\omega_1)S_j(\omega_2)g(\omega_1, \omega_2)\] for some specific $g$.

Let $\alpha$, $\beta$, and $\gamma$ be distinct points in $\Omega$. An application of Assumption~\ref{as:2} shows that
$g(\alpha, \beta) = C(\delta\pr{\omega\mid\alpha}, \delta\pr{\omega\mid\beta}) = C(\delta\pr{\omega\mid\alpha},\delta\pr{\omega\mid\gamma}) = g(\alpha, \gamma)$. Thus $g$ has a form $g(\omega_1, \omega_2) = a\delta\pr{\omega_1\mid \omega_2}+b$ for some constants $a$ and $b$.

To complete the proof note that Assumption~\ref{as:1} implies that $C(S+b) = C(S)$ which in turns implies that $\sum_x g(\omega_1, \omega_2) = 0$ for all $y$. Thus $b = - a\abs{\Omega}^{-1}$. This leads us to
\[
\begin{split}
C(S_i, S_j) & = \sum_{\omega_1,\omega_2 \in \Omega}S_i(\omega_1)S_j(\omega_2)\pr{a\delta\pr{\omega_1\mid\omega_2}-a\abs{\Omega}^{-1}} \\
& = a\sum_{\omega \in \Omega} S_i(\omega)S_j(\omega) - a\pr{\sum_{\omega \in \Omega} S_i(\omega)}\pr{\sum_{\omega \in \Omega} \abs{\Omega}^{-1}S_j(\omega)} \\
& \propto \mean{}{S_iS_j} - \mean{}{S_i}\mean{}{S_j},
\end{split}
\]
where the means are taken with respect to the uniform distribution. This identity proves the theorem.
\subsection{Proof for Lemma~\ref{lem:paritycov}}
Let us prove that $\cov{}{}{T_\iset{F}} = 0.5I$. Let $A$ be an itemset. There are odd number of ones in $A$ in exactly half of the transactions. Hence, $\mean{}{T_A^2} = \mean{}{T_A} = 0.5$. Let $B \neq A$ be an itemset. We wish to have $T_B(\omega) = T_A(\omega) = 1$. This means that $\omega$ must have odd number of ones in $A$ and in $B$. Assume that the number of ones in $A \cap B$ is even. This means that $A-B$ and $B-A$ have odd number of ones. There is only a quarter of all the transactions that fulfil this condition. If $A \cap B$ is odd, then we must an even number of ones in $A-B$ and $B-A$. Again, there is only a quarter of all the transactions for which this holds. This implies that $\mean{}{T_AT_B} = 0.25 = \mean{}{T_A}\mean{}{T_B}$. This proves that $\cov{}{}{T_\iset{F}} = 0.5I$.

\subsection{Proof of Theorem~\ref{thr:generic}}
\label{prf:generic}
Before proving this theorem let us rephrase it. First, note even though $\dista{\cdot}{\cdot}{\cdot}$ is defined only on the conjunction functions $S_\iset{F}$, we can operate with the parity function $T_\iset{F}$. As we stated before there is an invertible matrix $A$ such that $T_\iset{F} = AS_\iset{F}$. We can write the distance as
\[
\dista{D_1}{D_2}{S_\iset{F}}^2 = \pr{A\theta_1-A\theta_2}^T\pr{A^{-1}}^TC(S_\iset{F})^{-1}A^{-1}\pr{A\theta_1-A\theta_2}.
\]
Thus we define $C(T_\iset{F}) = AC(S_\iset{F})A^T$. Note that the following lemma implies that the condition stated in Theorem~\ref{thr:generic} is equivalent to $C(T_\iset{A}) = cI$, for some constant $c$. Theorem~\ref{thr:generic} is equivalent to stating that $C(T_\iset{F}) = cI$.

The following lemma deals with some difficulties due the fact that the frequencies should arise from some valid distributions
\begin{lemma}
\label{lem:valid}
Let $\iset{A}$ be the family of all itemsets. There exists $\epsilon > 0$ such that for each real vector $\gamma$ of length $2^K-1$ that satisfies $\norm{\gamma} < \epsilon$ there exist distributions $p$ and $q$ such that $\gamma = \mean{p}{T_\iset{A}}-\mean{q}{T_\iset{A}}$.
\end{lemma}
\begin{proof}
To ease the notation, add $T_0(x) = 1$ to $T_\iset{A}$ and denote the end result by $T^*$. We can consider $T^*$ as a $2^K \times 2^K$ matrix, say $A$. Let $p$ be a distribution and let $u$ be the vector of length $2^K$ representing the distribution. Note that we have $Au = \mean{p}{T^*}$. We can show that $A$ is invertible. Let $U$ some $2^K - 1$ dimensional open ball of distributions. Since $A$ is invertible, the set $V^* = \set{Ax \mid x \in U}$ is a $2^K - 1$ dimensional open ellipsoid. Define also $V$ by removing the first coordinate from the vectors of $V^*$. Note that the first coordinate of elements of $V^*$ is equal to $1$. This implies that $V$ is also a $2^K - 1$ dimensional open ellipsoid. Hence we can pick an open ball $N(\theta, \epsilon) \subset V$. The lemma follows.
\end{proof}
We are now ready to prove Theorem~\ref{thr:generic}:

Abbreviate the matrix $C(T_\iset{F})$ by $C$. We will first prove that the diagonal entries of $C$ are equal to $c$. Let $\iset{A}$ be the family of all itemsets. Select $G \in \iset{F}$ and define $\iset{R} = \set{H \in \iset{F} \mid H \subseteq G}$. As we stated above, $C(T_\iset{A}) = cI$ and Assumption~\ref{as2:3} imply that $C(T_\iset{R}) = cI$. Assumption~\ref{as2:2} implies that
\begin{equation}
\label{eq:chain}
\dista{\cdot}{\cdot}{S_\iset{R}}^2 \leq \dista{\cdot}{\cdot}{S_\iset{F}}^2 \leq \dista{\cdot}{\cdot}{S_\iset{A}}^2
\end{equation}
Select $\epsilon$ corresponding to Lemma~\ref{lem:valid} and let $\gamma_\iset{A} = \vect{0,\ldots,\epsilon/2,\ldots,0}$, i.e., $\gamma_\iset{A}$ is a vector whose entries are all $0$ except the entry corresponding to $G$. Lemma~\ref{lem:valid} guarantees that there exist distributions $p$ and $q$ such that $\dista{p}{q}{S_\iset{A}}^2 = c\norm{\gamma_\iset{A}}^2$. Let $\gamma_\iset{F} = \mean{p}{T_\iset{F}}-\mean{q}{T_\iset{F}}$ and $\gamma_\iset{R} = \mean{p}{T_\iset{R}}-\mean{q}{T_\iset{R}}$. Note that $\gamma_\iset{R}$ and $\gamma_\iset{F}$ has the same form as $\gamma_{A}$. It follows from Eq.~\ref{eq:chain} that
\[
c\epsilon^2/4 \leq C_{G,G}\epsilon^2/4 \leq c\epsilon^2/4,
\]
where $C_{G,G}$ is the diagonal entry of $C$ corresponding to $G$. It follows that $C_{G,G} = c$.

To complete the proof we need to show that $C_{G,H} = 0$ for $G,H\in \iset{F}, G \neq H$. Assume that $C_{X,Y} \neq 0$ and let $s$ be the sign of $C_{G,H}$. Apply Lemma~\ref{lem:valid} again and select $\gamma_\iset{A} = \spr{0,\ldots,\epsilon/4,0,\ldots,0,s\epsilon/4,\ldots,0}^T$, i.e., $\gamma_\iset{A}$ has $\epsilon/4$ and $s\epsilon/4$ in the entries corresponding to $G$ and $H$, respectively, and $0$ elsewhere. The right side of Eq.~\ref{eq:chain} implies that
\[
2c\epsilon^2/16+2\abs{C_{G,H}}\epsilon^2/16 \leq 2c\epsilon^2/16
\]
which is a contradiction and it follows that $C_{G,H} = 0$. This completes the theorem.

%% file: cm.bbl
\begin{thebibliography}{18}
\providecommand{\natexlab}[1]{#1}
\providecommand{\url}[1]{\texttt{#1}}
\expandafter\ifx\csname urlstyle\endcsname\relax
  \providecommand{\doi}[1]{doi: #1}\else
  \providecommand{\doi}{doi: \begingroup \urlstyle{rm}\Url}\fi

\bibitem[Agrawal et~al.(1993)Agrawal, Imielinski, and Swami]{agrawal93mining}
Rakesh Agrawal, Tomasz Imielinski, and Arun~N. Swami.
\newblock Mining association rules between sets of items in large databases.
\newblock In Peter Buneman and Sushil Jajodia, editors, \emph{Proceedings of
  the 1993 {ACM} {SIGMOD} International Conference on Management of Data},
  pages 207--216, Washington, D.C., 26--28~ 1993.

\bibitem[Agrawal et~al.(1996)Agrawal, Mannila, Srikant, Toivonen, and
  Verkamo]{agrawal96apriori}
Rakesh Agrawal, Heikki Mannila, Ramakrishnan Srikant, Hannu Toivonen, and
  Aino~I. Verkamo.
\newblock Fast discovery of association rules.
\newblock In Usama~M. Fayyad, Gregory Piatetsky-Shapiro, Padhraic Smyth, and
  Ramasamy Uthurusamy, editors, \emph{Advances in Knowledge Discovery and Data
  Mining}, pages 307--328. AAAI Press/The MIT Press, 1996.

\bibitem[Baldi et~al.(2003)Baldi, Frasconi, and Smyth]{baldi03internet}
Pierre Baldi, Paolo Frasconi, and Padhraic Smyth.
\newblock \emph{Modeling the Internet and the Web}.
\newblock John Wiley \& Sons, 2003.

\bibitem[Baseville(1989)]{baseville89distance}
Mich\'ele Baseville.
\newblock Distance measures for signal processing and pattern recognition.
\newblock \emph{Signal Processing}, 18\penalty0 (4):\penalty0 349--369, 1989.

\bibitem[Calders(2003)]{calders03thesis}
Toon Calders.
\newblock \emph{Axiomatization and Deduction Rules for the Frequency of
  Itemsets}.
\newblock PhD thesis, University of Antwerp, Belgium, 2003.

\bibitem[Calinski and Harabasz(1974)]{calinski74index}
Tadeusz Calinski and Jerzy Harabasz.
\newblock A dendrite method for cluster analysis.
\newblock \emph{Communications in Statistics}, 3:\penalty0 1--27, 1974.

\bibitem[Cooper(1990)]{cooper90complexity}
Gregory Cooper.
\newblock The computational complexity of probabilistic inference using
  bayesian belief networks.
\newblock \emph{Artificial Intelligence}, 42\penalty0 (2--3):\penalty0
  393--405, Mar. 1990.

\bibitem[Csisz\'ar(1975)]{csiszar75divergence}
Imre Csisz\'ar.
\newblock I-divergence geometry of probability distributions and minimization
  problems.
\newblock \emph{The Annals of Probability}, 3\penalty0 (1):\penalty0 146--158,
  Feb. 1975.

\bibitem[Davies and Bouldin(1979)]{davies79index}
David~L. Davies and Donald~W. Bouldin.
\newblock A cluster separation measure.
\newblock \emph{IEEE Transactions of Pattern Analysis and Machine
  Intelligence}, 1\penalty0 (2):\penalty0 224--227, April 1979.

\bibitem[Eiter and Mannila(1997)]{eiter97distance}
Thomas Eiter and Heikki Mannila.
\newblock Distance measures for point sets and their computation.
\newblock \emph{Acta Informatica}, 34\penalty0 (2):\penalty0 109--133, 1997.

\bibitem[Hailperin(1965)]{hailperin65inequalities}
Theodore Hailperin.
\newblock Best possible inequalities for the probability of a logical function
  of events.
\newblock \emph{The American Mathematical Monthly}, 72\penalty0 (4):\penalty0
  343--359, Apr. 1965.

\bibitem[Hand et~al.(2001)Hand, Mannila, and Smyth]{hand02principles}
David Hand, Heikki Mannila, and Padhraic Smyth.
\newblock \emph{Principles of Data Mining}.
\newblock The MIT Press, 2001.

\bibitem[Hollm\'en et~al.(2003)Hollm\'en, Sepp\"anen, and
  Mannila]{hollmen03mixture}
Jaakko Hollm\'en, Jouni~K Sepp\"anen, and Heikki Mannila.
\newblock Mixture models and frequent sets: combining global and local methods
  for 0-1 data.
\newblock In \emph{Proceedings of the SIAM Conference on Data Mining (2003)},
  2003.

\bibitem[Kullback(1968)]{kullback68information}
Solomon Kullback.
\newblock \emph{Information Theory and Statistics}.
\newblock Dover Publications, Inc., 1968.

\bibitem[Levenshtein(1966)]{levenshtein66distance}
Vladimir~I. Levenshtein.
\newblock Binary codes capable of correcting deletions, insertions and
  reversals.
\newblock \emph{Soviet Physics Doklady.}, 10\penalty0 (8):\penalty0 707--710,
  February 1966.

\bibitem[Mannila et~al.(1999)Mannila, Pavlov, and Smyth]{mannila99prediction}
Heikki Mannila, Dmitry Pavlov, and Padhraic Smyth.
\newblock Prediction with local patterns using cross-entropy.
\newblock In \emph{Knowledge Discovery and Data Mining}, pages 357--361, 1999.

\bibitem[Ng et~al.(2002)Ng, Jordan, and Weiss]{ng02clustering}
Andrew~Y. Ng, Michael~I. Jordan, and Yair Weiss.
\newblock On spectral clustering: Analysis and an algorithm.
\newblock In \emph{Advances in Neural Information Processing Systems 14}, 2002.

\bibitem[Pearl(1988)]{pearl88reasoning}
Judea Pearl.
\newblock \emph{Probabilistic reasoning in intelligent systems: networks of
  plausible inference}.
\newblock Morgan Kaufmann Publishers Inc., San Francisco, CA, USA, 1988.

\end{thebibliography}
